\documentclass[a4paper]{article}

\usepackage[english]{babel}
\usepackage[utf8]{inputenc}
\usepackage{amsmath}
\usepackage{graphicx}
\usepackage{epsfig}
\usepackage[colorinlistoftodos]{todonotes}
\usepackage[hidelinks]{hyperref}
\usepackage[margin=1.2in]{geometry}
\usepackage{natbib}
\usepackage{breqn}
\usepackage{bm}
\usepackage{bbm}
\usepackage{amssymb}
\usepackage{calc}

\usepackage{amsthm}
\usepackage{amsmath}               
  {
      \theoremstyle{plain}
      \newtheorem{assumption}{Assumption}
  }

\newtheorem{theorem}{Theorem}[section]

\usepackage{enumitem}

\title{Semiparametric Estimation on Multi-treatment Causal Effects via Cross-Fitting}

\author{Jingying Zeng}

\date{August 2021}

\begin{document}
\maketitle

\begin{abstract}
Causal inference is a critical research area with multi-disciplinary origins and applications, ranging from statistics, computer science, economics, psychology to public health. In many scientific research, randomized experiments provide a golden standard for estimation of causal effects for decades. However, in many situations, randomized experiments are not feasible in practice so that practitioners need to rely on empirical investigation for causal reasoning. Causal inference via observational data is a challenging task since the knowledge of the treatment assignment mechanism is missing, which typically requires non-testable assumptions to make the inference possible. For several years, great effort has been devoted to the research of causal inference for binary treatments. In practice, it is also common to use observational data on multiple treatment comparisons. Within the potential outcomes framework, we propose a generalized cross-fitting estimator (GCF), which generalizes the doubly robust estimator with cross-fitting for binary treatment to multiple treatment comparisons and provides rigorous proofs on its statistical properties. This estimator permits the use of more flexible machine learning methods to model the nuisance parts, and based on relatively weak assumptions, while there is still a theoretical guarantee for valid statistical inference. We show the asymptotic properties of the GCF estimators, and provide the asymptotic simultaneous confidence intervals that achieve the semiparametric efficiency bound for average treatment effect. The performance of the estimator is accessed through simulation study based on the common evaluation metrics generally considered in the causal inference literature.

\end{abstract}

\section{Introduction}
Even though randomized experiments are the gold standard in establishing causal links, under certain circumstances, randomized experiments are not ethical nor feasible so that practitioners rely on observational studies to make causal inference decisions. For instance, in clinical trials, it might not be ethical to give control group placebo when a new therapy is expected to prolong the survival rates of patients. In other situations, randomized experiments are not practical due to financial and time constrains or expectation of low compliance. For observational data, a variety of approaches including matching, stratification, and weighting have been extensively discussed in binary treatments settings. Non-binary treatment comparisons, however, is also very common in practice. For instance, researchers might be interested in comparing the treatment effects among multiple medications; or even continuous treatments like dose-response.

The weak unconfoundedness proposed by \citet{imbens2000role}, that is, within each strata defined by $\bm{X}$, each unit being assigned to treatment level $j$ versus not assigned to this level is randomized, brought the classical potential outcomes framework \citep{rosenbaum1983central} for multi-treatment comparison to light. Based on the generalized IPW estimator proposed by \citet{imbens2000role}, \citet{feng2012generalized} revised the estimator by normalizing the weights (which are generalized propensity scores) so that the sum of the weights for each treatment group is one in expectation. \citet{tu2013using} extended the AIPW estimator to a multi-treatment setting. Later on, a series of articles in statistics literature have discussed estimating the pairwise average treatment effect for multi-level comparisons through weighting approach \citep{lopez2017estimation}. Several matching methods are mentioned in the literature as well, for instance, \citet{rassen2013matching} discussed the three treatment groups comparison using propensity score matching. \citet{yang2016propensity} further extended the matching algorithm, along with the trimming rule to multi-treatment settings. A number of researchers have recognized the overlap issues that were exacerbated in multi-treatment comparisons. \citet{li2019propensity} proposed overlap weights to prevent trimming due to extreme estimated propensity scores and to achieve covariate balancing.

In this paper, we propose a generalized doubly-robust cross-fit (GCF) estimation method for multiple treatment comparisons and derive its corresponding asymptotic properties including simultaneous confidence intervals, which generalizes the cross-fit estimators for binary treatments \citep{chernozhukov2018double, newey2018cross} that can be used in conjunction with more flexible machine learning techniques to fit nuisance parameters. The GCF estimator has doubly robust property that gives double protection on model misspecification, and converges to the true causal effects at $\sqrt{n}$-convergence rate. Also, the asymptotic performance of the derived simultaneous confidence intervals is optimal.

%\citet{yang2016propensity} showed that when the assignment mechanism is weakly unconfounded, it is sufficient to estimate $\mathbb{E}\big[ Y(j) \big]$.

\section{Methodology}
\subsection{Model Specification and Assumptions}
The statistical settings for multi-treatment framework is adopted from \citet{imbens2000role}, which generalizes the causal effects defined through potential outcomes model \citep{neyman1923application, rubin1974estimating}. Suppose that $Z_i \in \mathcal{S} = \{ 1,...,J \}$, $i=1,...,N$,  denotes treatment level for each unit $i$. Each unit is associated with $J$ potential outcomes, $Y_i(j)$, and an indicator vector $\bm{D}_i = (D_{i1},...,D_{iJ})$, where $D_{ij} =  \mathbbm{1}\{ Z_i= j \}$. The observed data is $\big(Y_i^{obs}, Z_i, \bm{X}_i\big)\in \mathcal{R} \times \{0,1\} \times \mathcal{X}$, where $\bm{X}_i=(X_{i1},...,X_{ip})^T$ is a $p$-dimensional covariates and $Y_i^{obs}=Y_i(Z_i)$ is the observed outcome of interests by assuming Stable Unit Treatment Value Assumption (SUTVA). In general, SUTVA is a fundamental assumption required for potential outcomes framework, which assumes consistency and no interference such that $Y_i^{obs}=Y_i(\bm{Z})$ and $Y_i(\bm{z})=Y_i(z_i)$, where $\bm{Z}=(Z_1,...,Z_n)^T$ is the column vector of treatment assignment. \citet{imbens2000role} also defines the generalized propensity score (GPS): 
\begin{equation*}
    e_j(\bm{x}) = P(Z_i=j|\bm{X}_i=\bm{x}),
\end{equation*}
where $\sum_{j=1}^{J}e_j(\bm{x})=1$ for each unit. 

In order to identify causal effects, the following two assumptions are made \citep{rosenbaum1983central, imbens2000role}:

\begin{assumption}\label{as:a1}
Positivity (Overlap): $\xi < \mathbb{P}(Z_i=j|\bm{X}) < 1-\xi$, for some $\xi>0$ and for all $\forall$ $j$.
\end{assumption}

\begin{assumption}\label{as:a2}
Weak unconfoundedness: $Y(j) \perp \mathbbm{1} \{Z=j\} |\bm{X}$, for all $j \in \mathcal{S}$
\end{assumption}

%\ref{as:a1}

The average treatment effect (ATE) between treatment $j$ and $j'$ are defined as:
\begin{equation}\label{rawATE}
    \tau_{ATE_{j,j'}} \equiv \mathbb{E}[Y_i(j)-Y_i(j')], \text{ for } j \neq j' \text{ and } j,j' \in \{1,...,J\},
\end{equation}

\citet{imbens2000role} shows that under the assumption of SUTVA, positivity, and weak unconfoundedness, we have the following equation:
\begin{equation*}
    \mathbb{E}\big[ Y_i(j)|\bm{X} \big] = \mathbb{E}\big[Y_i^{obs}| Z_i=j, \bm{X} \big],
\end{equation*}
and hence the ATE between treatment $j$ and $j'$ in Equation (\ref{rawATE}) can be re-written as: 
\begin{equation}\label{ATE}
   \tau_{ATE_{j,j'}} = \tau_{ATE_j} - \tau_{ATE_{j'}}, \text{ for } j \neq j' \text{ and } j,j' \in \{1,...,J\},
\end{equation} 
where $\tau_{ATE_j} = \mathbb{E}_{\bm{X}}\big\{ \mathbb{E}[Y_i^{obs}|Z_i=j, \bm{X}] \big\}$. Intuitively, for treatment group $s$ such that $s \neq j$ or $j'$, $Y(s)$ plays no role in estimating the average treatment effect between treatments $j$ and $j'$, which is essentially the concept of Missing At Random (MAR) \citep{rubin1976inference}. In the following, we are going to propose a doubly robust estimator  of the ATE in (\ref{ATE}) with cross-fitting in the multiple treatment framework. 

The samples $\Omega=\{1,...,N\}$ can be partitioned into $K$ folds $E=\{I_1,...,I_k\}$ such that $\cup_{k=1}^{K}I_k=\Omega$, $I_k \cap I_{k'} = \emptyset$ when $k \neq k'$. For $k=1,...,K$,  $I_k^c= \Omega \setminus I_k$. The ATE between treatment $j$ and $j'$ can be estimated by using the generalized cross-fitting estimator (GCF) generalized from the doubly robust estimator with cross-fitting for binary treatment:

\begin{equation}\label{GDCDR}
\hat{\tau}_{GCF, j,j'} = \sum_{k=1}^{K} \frac{|I_k|}{N} \hat{\tau}^{I_k}_{j,j'}, \text{ for } j \neq j' \text{ and } j,j' \in \{1,...,J\},
\end{equation}
where
\begin{align}\label{tauIK}
\hat{\tau}^{I_k}_{j,j'} = &\frac{1}{|I_k|} \sum_{i \in I_k} \Bigg[  \hat{\mu}_j^{I_k^c}(\bm{X}_i) - \hat{\mu}_{j'}^{I_k^c}(\bm{X}_i)
+ 
\mathbbm{1}\{ Z_i=j\}  \frac{Y_i^{obs}-\hat{\mu}_j^{I_k^c}(\bm{X}_i) }{\hat{e}_j^{I_k^c}(\bm{X}_i)} \nonumber \\ 
 &- \mathbbm{1}\{ Z_i=j' \}  \frac{Y_i^{obs} - \hat{\mu}_{j'}^{I_k^c}(\bm{X}_i) }{1- \hat{e}_{j'}^{I_k^c}(\bm{X}_i)} \Bigg]. 
\end{align}
Here $\hat{\mu}_j^{I_k^c}(\bm{X}_i)$ and $\hat{e}_j^{I_k^c}(\bm{X}_i)$ denote the estimated outcome regression and propensity scores correspondingly excluding the k-th fold. When $|I_k|=\frac{N}{K}$, the estimator $\hat{\tau}_{GCF, j,j'}$ in equation (\ref{GDCDR}) can be simplified as $\hat{\tau}_{GCF, j,j'} = \frac{1}{K}  \sum_{k=1}^{K} \hat{\tau}^{I_k}$. Equivalently, the GCF estimator $\hat{\tau}_{GCF, j,j'}$ can be re-written as 
\begin{equation*}
\hat{\tau}_{GCF, j,j'} = \hat{\tau}_{GCF, j} - \hat{\tau}_{GCF, j'},
\end{equation*}
where $\hat{\tau}_{GCF, j}= \sum_{k=1}^{K} \frac{|I_k|}{N} \hat{\tau}^{I_k}_{j}$ and $\hat{\tau}^{I_k}_{j}= \frac{1}{|I_k|} \sum_{i \in I_k} \Big[  \hat{\mu}_j^{I_k^c}(\bm{X}_i) +
\mathbbm{1}\{ Z_i=j\}  \frac{Y_i^{obs}-\hat{\mu}_j^{I_k^c}(\bm{X}_i) }{\hat{e}_j^{I_k^c}(\bm{X}_i)} \Big]$.

The concept of cross-fitting was originally used in estimating nuisance parameters through i.i.d subsamples \citep{hajek1962asymptotically, chernozhukov2018double, newey2018cross}, which is similar to the idea of cross-validation \citep{zheng2011cross} in machine learning literature. Cross-fitting plays an important role in constructing asymptotic linear expansion \citep{van2000asymptotic} of the GCF estimator so that the semiparametric inference theories can be applied to prove its asymptotic normality, and the asymptotic confidence intervals can be constructed based on this result. Theoretical semiparametric inference gives elegant theories in analyzing the asymptotic behaviors of regular estimators that can be written as an asymptotic linear expansion as follows \citet{tsiatis2007semiparametric}:
\begin{equation}\label{asymptotic_linear}
    \sqrt{n}(\hat{\xi} - \xi) = \frac{1}{\sqrt{n}} \sum_{i=1}^{n} \psi_{\xi} (\bm{O}_i) + R_n,
\end{equation}
where $R_n$ is the residual term that converges to zero in probability, and $\psi_{\xi}(\mathcal{P})$ is the influence function with zero mean and finite variance. The functional form of nuisance parameter $\xi(\mathcal{P})$ is computed based on i.i.d sample $\bm{O}_1,...,\bm{O}_n$ from a distribution $\mathcal{P}$. When a regular estimator admits such asymptotic linearity, the results of asymptotic normality, asymptotic efficiency follows. In the following, we are going to demonstrate how cross-fitting contributes in establishing asymptotic linear expansion when more flexible models are used in estimating nuisance parameters.

Without using cross-fitting, the generalized AIPW estimator $\hat{\tau}_{GAIPW,j,j'}$ \citep{tu2013using} is defined as:
\begin{equation}\label{GAIPW}
\hat{\tau}_{GAIPW,j,j'} = \hat{\tau}_{GAIPW,j} - \hat{\tau}_{GAIPW,j'},
\end{equation}
where $\hat{\tau}_{GAIPW,j} = \frac{1}{N}\sum_{i=1}^{N}\Big[ \hat{\mu}_j(\bm{X}_i) + \mathbbm{1}\{ Z_i=j\}  \frac{Y_i^{obs}-\hat{\mu}_j(\bm{X}_i)   }{ \hat{e}_j(\bm{X}_i) }  \Big]$. The residual term $R_N$ of the estimator $\hat{\tau}_{GAIPW,j}$ is
\begin{equation}\label{residual}
R_N = \frac{1}{\sqrt{N}}\sum_{i=1}^{N}\Big[ g_j\Big(\bm{O}_i;\hat{\mu}_j(\bm{X}_i), \hat{e}_j(\bm{X}_i) \Big) - \hat{g}_j\Big(\bm{O}_i; \mu_j(\bm{X}_i), e_j(\bm{X}_i) \Big)  \Big],
\end{equation}
where $g_j\Big(\bm{O}_i; \mu_j(\bm{X}_i), e_j(\bm{X}_i) \Big)=\mu_j (\bm{X}_i) + \mathbbm{1}\{ Z_i=j\}  \frac{Y_i^{obs}-\mu_j(\bm{X}_i)   }{ e_j(\bm{X}_i) }$ and $ \hat{g}_j\Big(\bm{O}_i;\hat{\mu}_j(\bm{X}_i), \hat{e}_j(\bm{X}_i) \Big)$ is the plug-in estimator using $\hat{\mu}_j(\bm{X}_i)$ and $\hat{e}_j(\bm{X}_i)$. The asymptotic results followed by asymptotic linear expansion can not be directly applied to $\hat{\tau}_{GAIPW,j,j'}$ estimator since it is not a sum of i.i.d terms in Equation (\ref{residual}). One can make assumptions on the complexity of the nuisance parameters $\hat{\mu}_j(\cdot)$ and $\hat{e}_j(\cdot)$ to be within Donsker class \citep{van2000asymptotic} to satisfy asymptotic linearity. However, for flexible machine learning models, it can be hard to prove the model is from Donsker class. As an alternative, cross-fitting technique provides transparent assumptions and simplifies the proof of asymptotic linearity. The GCF estimator estimates the nuisance parameters on the out-of-fold dataset, and by conditioning on that, the summand becomes i.i.d terms  \citep{chernozhukov2018double, newey2018cross, zivich2021machine}. Therefore, the use of cross-fitting allows the use of more flexible models in estimating nuisance parameters and still have theoretical guarantees in these asymptotic results. The detailed assumptions and asymptotic properties of the estimator are discussed in Section 3.2.2.

%\begin{align}
%\hat{\tau}_{GAIPW,j,j'} &= \hat{\tau}_{GAIPW,j} - \hat{\tau}_{GAIPW,j'} \nonumber\\
%    & = \frac{1}{N}\sum_{i=1}^{N}\Big[ \hat{\mu}_j(\bm{X}_i) + \mathbbm{1}\{ Z_i=j\}  \frac{Y_i^{obs}-\hat{\mu}_j(\bm{X}_i)   }{ \hat{e}_j(\bm{X}_i) }  \Big] - \frac{1}{N}\sum_{i=1}^{N}\Big[ \hat{\mu}_{j'}(\bm{X}_i) + \mathbbm{1}\{ Z_i=j'\}  \frac{Y_i^{obs}-\hat{\mu}_{j'}(\bm{X}_i)   }{ \hat{e}_{j'}(\bm{X}_i) }  \Big]
%\end{align}

Furthermore, we propose the $(1-\alpha)$ Wald-type \citep{yang2016propensity} simultaneous confidence interval for the GCF estimator using Bonferroni correction \citep{dunn1961multiple}, defined as follows:
\begin{equation}\label{eq:var_wald}
\Bigg[ \hat{\tau}_{GCF, j,j'}  \pm \frac{\phi^{-1}(1-\frac{\alpha}{2 \cdot {J \choose 2}})}{\sqrt{N}}\sqrt{\hat{V}ar(\hat{\tau}_{GCF, j,j'})} \Bigg], \text{ for } j \neq j' \text{ and } j,j' \in \{1,...,J\},
\end{equation}  
where
\begin{equation}
\begin{split}
\hat{V}ar(\hat{\tau}_{GCF, j,j'}) &= \frac{1}{N-1}\sum_{k=1}^K\sum_{i \in I_k}\Bigg( \hat{\mu}_j^{I_k^c}(\bm{X}_i) - \hat{\mu}_{j'}^{I_k^c}(\bm{X}_i)  \nonumber\\
& + 
\mathbbm{1}\{ Z_i=j\}  \frac{Y_i^{obs}-\hat{\mu}_j^{I_k^c}(\bm{X}_i) }{\hat{e}_j^{I_k^c}(\bm{X}_i)} 
 - \mathbbm{1}\{ Z_i=j' \}  \frac{Y_i^{obs} - \hat{\mu}_{j'}^{I_k^c}(\bm{X}_i) }{1- \hat{e}_j^{I_k^c}(\bm{X}_i)} - \hat{\tau}_{GCF, j,j'}  \Bigg)^2. 
\end{split}
\end{equation}

%The asymptotic normality holds since the GCF estimator admits the asymptotic linear expansion, which implies its asymptotic normality. 
With asymptotic normality holds, one can use this optimal confidence interval to quantify the uncertainty of the GCF estimator. Next, we are going to show in Section 3.2.2 that under certain assumptions, this GCF estimator achieves good asymptotic results.

%\newpage
\subsection{Asymptotic Results for the GCF}
In this section, we exam the asymptotic behavior of the GCF estimator. The Theorem \ref{theorem:GDCDR} and Theorem \ref{theorem: var} show that under certain conditions, the GCF estimator is $\sqrt{n}$-consistent, asymptotic efficient, and asymptotic normal. Therefore, the confidence intervals constructed based on these asymptotic results are optimal.
\begin{theorem}\label{theorem:GDCDR}
With i.i.d samples and Assumptions \ref{as:a1}, \ref{as:a2}, if, $\forall$ $k=1,...,K$, and $\forall$ $j=1,...,J$,
\begin{enumerate}
    \item $\sup_{x \in \mathcal{X}}\big| \hat{\mu}_j^{I_k^c}(\bm{x}) - \mu_j(\bm{x}) \big| \overset{p}{\rightarrow}0$

     \item $\sup_{x \in \mathcal{X}}\big| \hat{e}_j^{I_k^c}(\bm{x}) - e_j(\bm{x}) \big| \overset{p}{\rightarrow}0$
     
     \item $\sqrt{n}\cdot MSE\big[ \hat{\mu}_j^{I_k^c}(\bm{x})  \big] \cdot MSE\big[ \hat{e}_j^{I_k^c}(\bm{x})  \big] \overset{p}{\rightarrow}0$
\end{enumerate}

then $\hat{\tau}_{GCF,j,j'}$ is $\sqrt{n}$-consistent.
\end{theorem}

\begin{proof}
First we see that to show $\sqrt{n}(\hat{\tau}_{GCF,j,j'} - \tau_{GCF, j,j'}) \overset{p}{\rightarrow}0$ is equivalent to show that
\begin{align}\label{lam:equ1}
\sqrt{n}\frac{1}{|I_k|} \sum_{i \in I_k}\Big[ \hat{\mu}_j^{I_k^c}(\bm{X}_i) &+ 
\mathbbm{1}\{ Z_i=j\}  \frac{Y_i^{obs}-\hat{\mu}_j^{I_k^c}(\bm{X}_i) }{\hat{e}_j^{I_k^c}(\bm{X}_i)} \nonumber \\ &- \mu_j(\bm{X}_i) - 
\mathbbm{1}\{ Z_i=j\}  \frac{Y_i^{obs}-\mu_j(\bm{X}_i) }{e_j(\bm{X}_i)}    \Big] \overset{p}{\rightarrow}0.
\end{align}
The terms on the left-hand side of (\ref{lam:equ1}) can be decomposed into three parts as follows.
\begin{align}\label{lam:equ1-1}
   \frac{1}{|I_k|} \sum_{i \in I_k} \Bigg[ \hat{\mu}_j^{I_k^c}(\bm{X}_i) &+ 
\mathbbm{1}\{ Z_i=j\}  \frac{Y_i^{obs}-\hat{\mu}_j^{I_k^c}(\bm{X}_i) }{\hat{e}_j^{I_k^c}(\bm{X}_i)} - \mu_j(\bm{X}_i) - 
\mathbbm{1}\{ Z_i=j\}  \frac{Y_i^{obs}-\mu_j(\bm{X}_i) }{e_j(\bm{X}_i)} \Bigg] \nonumber\\
    =&- \frac{1}{|I_k|} \sum_{i \in I_k} \Bigg[ \big( \hat{\mu}_j^{I_k^c}(\bm{X}_i) -\mu_j(\bm{X}_i) \big) \Bigg( \frac{\mathbbm{1}\{ Z_i=j\}}{\hat{e}_j^{I_k^c}(\bm{X}_i)}-\frac{\mathbbm{1}\{ Z_i=j\}}{e_j(\bm{X}_i)} \Bigg) \Bigg]\\\label{lam:equ1-2}
    &+\frac{1}{|I_k|} \sum_{i \in I_k} \Bigg[ \big( Y_i^{obs} -\mu_j(\bm{X}_i) \big) \Bigg( \frac{\mathbbm{1}\{ Z_i=j\}}{\hat{e}_j^{I_k^c}(\bm{X}_i)}-\frac{\mathbbm{1}\{ Z_i=j\}}{e_j(\bm{X}_i)} \Bigg) \Bigg] \\\label{lam:equ1-3}
     &+ \frac{1}{|I_k|} \sum_{i \in I_k} \Bigg[ \big( \hat{\mu}_j^{I_k^c}(\bm{X}_i) -\mu_j(\bm{X}_i) \big) \Bigg( 1-\frac{\mathbbm{1}\{ Z_i=j\}}{e_j(\bm{X}_i)} \Bigg) \Bigg].
\end{align}
Now we are going to show that each of the three parts, (\ref{lam:equ1-1} -- \ref{lam:equ1-3}), on the right-hand side of equality above converges to zero in probability.

For (\ref{lam:equ1-1}), using Cauchy-Schwarz inequality, along with the assumptions in the theorem, as well as the positivity Assumption 1, we can show that %\key{(I still don't get it how you prove this.)}
\begin{flalign*}
\frac{1}{|I_k|} &\sum_{i \in I_k} \Bigg[ \big( \hat{\mu}_j^{I_k^c}(\bm{X}_i) -\mu_j(\bm{X}_i) \big) \Bigg( \frac{\mathbbm{1}\{ Z_i=j\}}{\hat{e}_j^{I_k^c}(\bm{X}_i)}-\frac{\mathbbm{1}\{ Z_i=j\}}{e_j(\bm{X}_i)} \Bigg) \Bigg]\\
&\leq \sqrt{\frac{1}{|I_k|} \sum_{i \in I_k} \big( \hat{\mu}_j^{I_k^c}(\bm{X}_i) -\mu_j(\bm{X}_i) \big)^2 } \cdot \sqrt{\frac{1}{|I_k|} \sum_{i \in I_k} \Bigg( \frac{\mathbbm{1}\{ Z_i=j\}}{\hat{e}_j^{I_k^c}(\bm{X}_i)}-\frac{\mathbbm{1}\{ Z_i=j\}}{e_j(\bm{X}_i)} \Bigg)^2 } \\
& \leq \sqrt{\frac{1}{|I_k|} \sum_{i \in I_k} \big( \hat{\mu}_j^{I_k^c}(\bm{X}_i) -\mu_j(\bm{X}_i) \big)^2 } \cdot \sqrt{\frac{1}{|I_k|} \sum_{i \in I_k}  \Bigg( \frac{e_j(\bm{X}_i) - \hat{e}_j^{I_k^c} }{\hat{e}_j^{I_k^c}(\bm{X}_i) \cdot e_j(\bm{X}_i) } \Bigg)^2 } \\
& \leq \sqrt{\frac{1}{|I_k|} \sum_{i \in I_k} \big( \hat{\mu}_j^{I_k^c}(\bm{X}_i) -\mu_j(\bm{X}_i) \big)^2 } \cdot \sqrt{\frac{1}{|I_k|} \sum_{i \in I_k}  \Bigg( \frac{1}{\xi^2} \Bigg)^2 \Bigg( e_j(\bm{X}_i) - \hat{e}_j^{I_k^c}(\bm{X}_i) \Bigg)^2}
= o_P\Bigg(\frac{1}{\sqrt{n}}\Bigg). 
\end{flalign*}
Note that $\hat{e}_j^{(-t(i))}(\bm{X}_i) > \xi$ for sufficiently large $n$ (by positivity Assumption 1 and the sup-norm consistency of the estimated propensity scores, i.e. $\sup_{x \in \mathcal{X}}\big| \hat{e}_j^{I_k^c}(\bm{x}) - e_j(\bm{x}) \big| \overset{p}{\rightarrow}0$). Additionally, $\frac{1}{|I_k|} \sum_{i \in I_k} \big( \hat{\mu}_j^{I_k^c}(\bm{X}_i) -\mu_j(\bm{X}_i) \big)^2$ and $\frac{1}{|I_k|} \sum_{i \in I_k} \big( \hat{e}_j^{I_k^c}(\bm{X}_i) - e_j^{I_k^c}(\bm{X}_i) \big)^2$ converge in probability to $MSE\big[\hat{\mu}_j^{I_k^c}(\bm{x}) \big]$ and $MSE\big[\hat{e}_j^{I_k^c}(\bm{x}) \big]$ correspondingly by law of large number. And because $\sqrt{MSE\big[\hat{\mu}_j^{I_k^c}(\bm{x})  \big] \cdot MSE\big[ \hat{e}_j^{I_k^c}(\bm{x})  \big]}=o_P\Big(\frac{1}{\sqrt{n}}\Big)$ by assumption c, the (\ref{lam:equ1-1}) is $o_P\Big(\frac{1}{\sqrt{n}}\Big)$.

To show that the (\ref{lam:equ1-2}) and (\ref{lam:equ1-3}) converge to zero in probability, first notice that when conditioning on $I_k^c$, both $\hat{e}_j^{I_k^c}(\bm{X}_i)$ and $\hat{\mu}_j^{I_k^c}(\bm{X}_i)$ are deterministic, or in other words, $\hat{e}_j^{I_k^c}(\bm{X}_i)$ or $\hat{\mu}_j^{I_k^c}(\bm{X}_i)$ will not introduce any dependency between samples in $I_k$, and hence we have
\begin{align}\label{eq:indep}
\big( Y_i^{obs} -\mu_j(\bm{X}_i) \big) \left( \frac{\mathbbm{1}\{ Z_i=j\}}{\hat{e}_j^{I_k^c}(\bm{X}_i)}-\frac{\mathbbm{1}\{ Z_i=j\}}{e_j(\bm{X}_i)} \right) &\text{ are i.i.d., } \forall i \in I_k,\nonumber \\
\big( \hat{\mu}_j^{I_k^c}(\bm{X}_i) -\mu_j(\bm{X}_i) \big) \left( 1-\frac{\mathbbm{1}\{ Z_i=j\}}{e_j(\bm{X}_i)} \right) &\text{ are i.i.d., } \forall i \in I_k. 
\end{align}
Additionally, we may need the following results:
\begin{equation}\label{eq:E1}
    \mathbb{E}\left[  \frac{1}{|I_k|} \sum_{i \in I_k}  \big( Y_i^{obs} -\mu_j(\bm{X}_i) \big) \left( \frac{\mathbbm{1}\{ Z_i=j\}}{\hat{e}_j^{I_k^c}(\bm{X}_i)}-\frac{\mathbbm{1}\{ Z_i=j\}}{e_j(\bm{X}_i)} \right)   \Bigg| I_k^c \right]=0,\end{equation} and 
\begin{equation}\label{eq:E2}
\mathbb{E}\left[ \frac{1}{|I_k|} \sum_{i \in I_k} \big( \hat{\mu}_j^{I_k^c}(\bm{X}_i) -\mu_j(\bm{X}_i) \big) \left( 1-\frac{\mathbbm{1}\{ Z_i=j\}}{e_j(\bm{X}_i)} \right) \Bigg| I_k^c \right]=0.
\end{equation}

Notice that since \begin{flalign*}
\mathbb{E}\Bigg[\big( Y_i^{obs} &-\mu_j(\bm{X}_i) \big) \big( \frac{\mathbbm{1}\{ Z_i=j\}}{\hat{e}_j^{I_k^c}(\bm{X}_i)}-\frac{\mathbbm{1}\{ Z_i=j\}}{e_j(\bm{X}_i)} \big) \big| \bm{X}_i, I_k^c \Bigg]\\
= &\Bigg[ \frac{\mathbb{E}(Z_i=j|\bm{X}_i, I_k^c)}{\hat{e}_j^{I_k^c}(\bm{X}_i)}-\frac{\mathbb{E}(Z_i=j|\bm{X}_i, I_k^c)}{e_j(\bm{X}_i)} \Bigg]\cdot \bigg[ \mathbb{E}(Y_i^{obs}|\bm{X}_i) - \mu_j(\bm{X}_i) \bigg]&\\
= &\Bigg[ \frac{\mathbb{E}(Z_i=j|\bm{X}_i, I_k^c)}{\hat{e}_j^{I_k^c}(\bm{X}_i)}-\frac{\mathbb{E}(Z_i=j|\bm{X}_i, I_k^c)}{e_j(\bm{X}_i)} \Bigg] \cdot \bigg[ \mathbb{E}\{Y_i(j)|Z_i=j,\bm{X}_i \} - \mu_j(\bm{X}_i) \bigg]\\
= &\Bigg[ \frac{\mathbb{E}(Z_i=j|\bm{X}_i, I_k^c)}{\hat{e}_j^{I_k^c}(\bm{X}_i)}-\frac{\mathbb{E}(Z_i=j|\bm{X}_i, I_k^c)}{e_j(\bm{X}_i)} \Bigg] \cdot 0 =0,
\end{flalign*} 
then (\ref{eq:E1}) can be written as
\begin{flalign*}
\mathbb{E}\Bigg[  \frac{1}{|I_k|} &\sum_{i \in I_k}  \bigg( Y_i^{obs} -\mu_j(\bm{X}_i) \bigg) \left( \frac{\mathbbm{1}\{ Z_i=j\}}{\hat{e}_j^{I_k^c}(\bm{X}_i)}-\frac{\mathbbm{1}\{ Z_i=j\}}{e_j(\bm{X}_i)} \right)   \Bigg| I_k^c \Bigg]\\
 = &\mathbb{E}\Bigg[ \mathbb{E}\bigg[  \frac{1}{|I_k|} \sum_{i \in I_k}  \bigg( Y_i^{obs} -\mu_j(\bm{X}_i) \bigg) \Bigg( \frac{\mathbbm{1}\{ Z_i=j\}}{\hat{e}_j^{I_k^c}(\bm{X}_i)}-\frac{\mathbbm{1}\{ Z_i=j\}}{e_j(\bm{X}_i)} \Bigg)   \Bigg| \bm{X}_i, I_k^c \bigg] \Bigg]&\\
= &\mathbb{E}\Bigg[   \frac{1}{|I_k|}  \sum_{i \in I_k} \mathbb{E}\bigg[ \bigg(Y_i^{obs} -\mu_j(\bm{X}_i) \bigg) \Bigg( \frac{\mathbbm{1}\{ Z_i=j\}}{\hat{e}_j^{I_k^c}(\bm{X}_i)}-\frac{\mathbbm{1}\{ Z_i=j\}}{e_j(\bm{X}_i)} \Bigg)   \Bigg| \bm{X}_i, I_k^c \bigg] \Bigg] =0.
\end{flalign*}
Similarly, since 
\begin{flalign*}
\mathbb{E}\Bigg[ &\bigg( \hat{\mu}_j^{I_k^c}(\bm{X}_i) -\mu_j(\bm{X}_i) \bigg) \Bigg( 1-\frac{\mathbbm{1}\{ Z_i=j\}}{e_j(\bm{X}_i)} \Bigg) \Bigg| \bm{X}_i, I_k^c \Bigg]\\
&= \bigg( \hat{\mu}_j^{I_k^c}(\bm{X}_i) -\mu_j(\bm{X}_i) \bigg)   \frac{e_j(\bm{X}_i)-\mathbb{E}\Big[ \mathbbm{1}\{ Z_i=j\} \big| \bm{X}_i, I_k^c \Big]}{e_j(\bm{X}_i)} \\
&= \bigg( \hat{\mu}_j^{I_k^c}(\bm{X}_i) -\mu_j(\bm{X}_i) \bigg)   \frac{e_j(\bm{X}_i)-\mathbb{P}(Z_i=j \big| \bm{X}_i, I_k^c)}{e_j(\bm{X}_i)} = \bigg( \hat{\mu}_j^{I_k^c}(\bm{X}_i) -\mu_j(\bm{X}_i) \bigg)\cdot 0=0,
\end{flalign*}
we derive
\begin{flalign*}
\mathbb{E}\Bigg[ &\frac{1}{|I_k|} \sum_{i \in I_k} \bigg( \hat{\mu}_j^{I_k^c}(\bm{X}_i) -\mu_j(\bm{X}_i) \bigg) \bigg( 1-\frac{\mathbbm{1}\{ Z_i=j\}}{e_j(\bm{X}_i)} \bigg) \Bigg| I_k^c \Bigg]\\
&=\mathbb{E}\Bigg[ \mathbb{E}\bigg[ \frac{1}{|I_k|} \sum_{i \in I_k} \bigg( \hat{\mu}_j^{I_k^c}(\bm{X}_i) -\mu_j(\bm{X}_i) \bigg) \Bigg( 1-\frac{\mathbbm{1}\{ Z_i=j\}}{e_j(\bm{X}_i)} \Bigg) \Bigg| \bm{X}_i, I_k^c \bigg] \Bigg]=0
\end{flalign*}
Hence, (\ref{eq:E2}) holds. 

Now, given $I_k^c$, since the mean of (\ref{lam:equ1-2}) is 0 as in (\ref{eq:E1}), so is its unconditional mean. The variance of (\ref{lam:equ1-2}) thus can be written as
\begin{flalign*}
\mathbb{E}\Bigg[ \Bigg\{  \frac{1}{|I_k|} &\sum_{i \in I_k}  \big( Y_i^{obs} -\mu_j(\bm{X}_i) \big) \Bigg( \frac{\mathbbm{1}\{ Z_i=j\}}{\hat{e}_j^{I_k^c}(\bm{X}_i)}-\frac{\mathbbm{1}\{ Z_i=j\}}{e_j(\bm{X}_i)} \Bigg)   \Bigg\}^2 \Bigg]\\
= &\mathbb{E}\Bigg[  \mathbb{E}\Bigg[ \Bigg\{  \frac{1}{|I_k|} \sum_{i \in I_k}  \big( Y_i^{obs} -\mu_j(\bm{X}_i) \big) \Bigg( \frac{\mathbbm{1}\{ Z_i=j\}}{\hat{e}_j^{I_k^c}(\bm{X}_i)}-\frac{\mathbbm{1}\{ Z_i=j\}}{e_j(\bm{X}_i)} \Bigg)   \Bigg\}^2\Bigg| I_k^c \Bigg] \Bigg]\\
= &\mathbb{E}\Bigg[  Var\Bigg[  \frac{1}{|I_k|} \sum_{i \in I_k}  \big( Y_i^{obs} -\mu_j(\bm{X}_i) \big) \Bigg( \frac{\mathbbm{1}\{ Z_i=j\}}{\hat{e}_j^{I_k^c}(\bm{X}_i)}-\frac{\mathbbm{1}\{ Z_i=j\}}{e_j(\bm{X}_i)} \Bigg)   \Bigg| I_k^c \Bigg]+0 \Bigg].
\end{flalign*}
Since the terms inside the sum above are independent, see (\ref{eq:indep}), the above can be expressed as
\begin{flalign*}
&\frac{1}{|I_k|} \mathbb{E}\Bigg[  Var\Bigg[    \big( Y_i^{obs} -\mu_j(\bm{X}_i) \big) \Bigg( \frac{\mathbbm{1}\{ Z_i=j\}}{\hat{e}_j^{I_k^c}(\bm{X}_i)}-\frac{\mathbbm{1}\{ Z_i=j\}}{e_j(\bm{X}_i)} \Bigg)   \Bigg| I_k^c \Bigg] \Bigg] \\
=&\frac{1}{|I_k|} \mathbb{E}\Bigg[  \mathbb{E}\Bigg[    \big( Y_i^{obs} -\mu_j(\bm{X}_i) \big)^2 \Bigg( \frac{\mathbbm{1}\{ Z_i=j\}}{\hat{e}_j^{I_k^c}(\bm{X}_i)}-\frac{\mathbbm{1}\{ Z_i=j\}}{e_j(\bm{X}_i)} \Bigg)^2   \Bigg| I_k^c \Bigg] - 0 \Bigg],
\end{flalign*}
for a particular $i$. 
Now continue investigating the double expectations above by conditioning on $\bm{X}_i$. It can be expressed as
\begin{flalign}\label{eq:E3}
&\frac{1}{|I_k|} \mathbb{E} \Bigg[ \mathbb{E}\Bigg[  \mathbb{E}\Bigg[    \big( Y_i^{obs} -\mu_j(\bm{X}_i) \big)^2 \Bigg( \frac{\mathbbm{1}\{ Z_i=j\}}{\hat{e}_j^{I_k^c}(\bm{X}_i)}-\frac{\mathbbm{1}\{ Z_i=j\}}{e_j(\bm{X}_i)} \Bigg)^2   \Bigg| \bm{X}_i,I_k^c \Bigg] \Bigg] \Bigg]\nonumber\\
%=&\frac{1}{|I_k|} \mathbb{E} \Bigg[ \mathbb{E}\Big[  \mathbb{E}\big[    \big(Y_i^{obs} -\mu_j(\bm{X}_i) \big)^2| \bm{X}_i \big] \cdot  \mathbb{E}\big[  \mathbbm{1}^2\{ Z_i=j\} \cdot \big( \frac{1}{\hat{e}_j^{I_k^c}(\bm{X}_i)}-\frac{1}{e_j(\bm{X}_i)} \big)^2   \big| \bm{X}_i,I_k^c \big] \Big] \Bigg]&\\
=&\frac{1}{|I_k|} \mathbb{E} \Bigg[ \mathbb{E}\Bigg[  Var(Y_i^{obs}|\bm{X}_i) \cdot  \mathbb{E}\Bigg[  \mathbbm{1}\{ Z_i=j\} \cdot \Bigg( \frac{1}{\hat{e}_j^{I_k^c}(\bm{X}_i)}-\frac{1}{e_j(\bm{X}_i)} \Bigg)^2   \Bigg| \bm{X}_i,I_k^c \Bigg] \Bigg] \Bigg]\nonumber\\
%&=\frac{1}{|I_k|} Var(Y_i^{obs}|\bm{X}_i) \cdot \mathbb{E} \Bigg[ \mathbb{E}\Big[ \mathbb{E}\big[  \mathbbm{1}\{ Z_i=j\} \cdot  \big(\frac{1}{\hat{e}_j^{I_k^c}(\bm{X}_i)}-\frac{1}{e_j(\bm{X}_i)} \big)^2   \big| \bm{X}_i,I_k^c \big] \Big] \Bigg]&\\
=&\frac{1}{|I_k|} Var(Y_i^{obs}|\bm{X}_i) \cdot \mathbb{E} \Bigg[ \mathbb{E}\Bigg[  \Bigg(\frac{1}{\hat{e}_j^{I_k^c} (\bm{X}_i)}-\frac{1}{e_j(\bm{X}_i)} \Bigg)^2 \cdot \mathbb{E}\big[  \mathbbm{1}\{ Z_i=j\}  \Bigg| \bm{X}_i,I_k^c \big] \Bigg] \Bigg]\nonumber\\
=&\frac{1}{|I_k|} Var(Y_i^{obs}|\bm{X}_i) \cdot \mathbb{E}\Bigg[  \Bigg(\frac{1}{\hat{e}_j^{I_k^c}(\bm{X}_i)}-\frac{1}{e_j(\bm{X}_i)} \Bigg)^2 \cdot e_j(\bm{X}_i) \Bigg].
\end{flalign}
Since $e_j(\bm{X}_i) \cdot \Bigg(\frac{1}{\hat{e}_j^{I_k^c}(\bm{X}_i)}-\frac{1}{e_j(\bm{X}_i)} \Bigg)^2   \leq \frac{1-\xi}{\xi \cdot \xi} \overset{p}{\rightarrow} 0$ by Assumption 1, and $I_k \sim n/K$, then
\begin{flalign*}
&\mathbb{E}\Bigg[ \Bigg\{  \frac{1}{|I_k|} \sum_{i \in I_k}  \big( Y_i^{obs} -\mu_j(\bm{X}_i) \big) \Bigg( \frac{\mathbbm{1}\{ Z_i=j\}}{\hat{e}_j^{I_k^c}(\bm{X}_i)}-\frac{\mathbbm{1}\{ Z_i=j\}}{e_j(\bm{X}_i)} \Bigg)   \Bigg\}^2 \Bigg]\\
=&\frac{1}{|I_k|} Var(Y_i^{obs}|\bm{X}_i) \cdot \mathbb{E} \Bigg[  \Bigg(\frac{1}{\hat{e}_j^{I_k^c}(\bm{X}_i)}-\frac{1}{e_j(\bm{X}_i)} \Bigg)^2 \cdot e_j(\bm{X}_i) \Bigg] = o_P(\frac{1}{n}).
%&= \frac{o(1)}{n} & (**)
\end{flalign*}
Consequently, the term in (\ref{lam:equ1-2}) converges to zero in probability.

The term of (\ref{lam:equ1-3}) converging to zero in probability can also be proved analogously as that for the term (\ref{lam:equ1-2}). Using (\ref{eq:E2}), the variance of (\ref{lam:equ1-3}) can be expressed, similarly as in (\ref{eq:E3}) 
\begin{flalign}\label{eq:E4}
&\mathbb{E}\big[ \{ \frac{1}{|I_k|} \sum_{i \in I_k}  \big( \hat{\mu}_j^{I_k^c}(\bm{X}_i) -\mu_j(\bm{X}_i) \big) \big( 1-\frac{\mathbbm{1}\{ Z_i=j\}}{e_j(\bm{X}_i)} \big)  \}^2 \big] \nonumber\\
%&=\mathbb{E}\Big[ \mathbb{E}\big[ \{ \frac{1}{|I_k|} \sum_{i \in I_k}  \big( \hat{\mu}_j^{I_k^c}(\bm{X}_i) -\mu_j(\bm{X}_i) \big) \big( 1-\frac{\mathbbm{1}\{ Z_i=j\}}{e_j(\bm{X}_i)} \big) \}^2 \big| I_k^c \big] \Big]&\\
%&=\mathbb{E}\Big[ Var\big[  \frac{1}{|I_k|} \sum_{i \in I_k}  \big( \hat{\mu}_j^{I_k^c}(\bm{X}_i) -\mu_j(\bm{X}_i) \big) \big( 1-\frac{\mathbbm{1}\{ Z_i=j\}}{e_j(\bm{X}_i)} \big)  \big| I_k^c \big] + 0 \Big]&\\
%&=\frac{1}{|I_k|}\mathbb{E}\Big[ Var\big[   \big( \hat{\mu}_j^{I_k^c}(\bm{X}_i) -\mu_j(\bm{X}_i) \big) \big( 1-\frac{\mathbbm{1}\{ Z_i=j\}}{e_j(\bm{X}_i)} \big)  \big| I_k^c \big] \Big]&\\
%&=\frac{1}{|I_k|}\mathbb{E}\Big[ \mathbb{E} \big[   \big( \hat{\mu}_j^{I_k^c}(\bm{X}_i) -\mu_j(\bm{X}_i) \big)^2 \big( 1-\frac{\mathbbm{1}\{ Z_i=j\}}{e_j(\bm{X}_i)} \big)^2  \big| I_k^c \big] + 0 \Big]&\\
= &\frac{1}{|I_k|}\mathbb{E}\Bigg[\mathbb{E} \Bigg[   \big( \hat{\mu}_j^{I_k^c}(\bm{X}_i) -\mu_j(\bm{X}_i) \big)^2 \Bigg( 1-\frac{\mathbbm{1}\{ Z_i=j\}}{e_j(\bm{X}_i)} \Bigg)^2  \Bigg|\bm{X}_i, I_k^c \Bigg]\Bigg] \nonumber\\
%&=\frac{1}{|I_k|}\mathbb{E}\bigg[\mathbb{E}\Big[   \big( \hat{\mu}_j^{I_k^c}(\bm{X}_i) -\mu_j(\bm{X}_i) \big)^2 \cdot \mathbb{E} \big[ \big( 1-\frac{\mathbbm{1}\{ Z_i=j\}}{e_j(\bm{X}_i)} \big)^2  \big|\bm{X}_i, I_k^c \big]\Big] \Big| I_k^c \bigg]&\\
= &\frac{1}{|I_k|}\mathbb{E}\Bigg[\mathbb{E}\Bigg[   \big( \hat{\mu}_j^{I_k^c}(\bm{X}_i) -\mu_j(\bm{X}_i) \big)^2  \cdot \mathbb{E} \Bigg[  1^2-2\frac{\mathbbm{1}\{ Z_i=j\}}{e_j(\bm{X}_i)} + \frac{\mathbbm{1}^2\{ Z_i=j\}}{e_j(\bm{X}_i)} \Bigg| \bm{X}_i, I_k^c \Bigg]\Bigg] \Bigg| I_k^c \Bigg].
\end{flalign}
The inside conditional expectation in (\ref{eq:E4}) can be shown as 
\begin{flalign*}
\mathbb{E} \Bigg[  1^2-2\frac{\mathbbm{1}\{ Z_i=j\}}{e_j(\bm{X}_i)} + \frac{\mathbbm{1}^2\{ Z_i=j\}}{e_j(\bm{X}_i)} \Bigg| \bm{X}_i, I_k^c \Bigg] &= 1-2\frac{\mathbb{P} \big[  Z_i=j \big| \bm{X}_i, I_k^c \big]}{e_j(\bm{X}_i)} + \frac{\mathbb{P} \big[ Z_i=j\big| \bm{X}_i, I_k^c \big]}{e^2_j(\bm{X}_i)} \\
&= -1+ \frac{1}{e_j(\bm{X}_i)} \le \frac{1}{\xi}.
\end{flalign*}
Hence, (\ref{eq:E4}) can be shown as
\begin{flalign*}
\mathbb{E}\Bigg[ \Bigg\{ \frac{1}{|I_k|} \sum_{i \in I_k}  \big( \hat{\mu}_j^{I_k^c}(\bm{X}_i) &-\mu_j(\bm{X}_i) \big) \Bigg( 1-\frac{\mathbbm{1}\{ Z_i=j\}}{e_j(\bm{X}_i)} \Bigg)  \Bigg\}^2 \Bigg] \\
%&=\frac{1}{|I_k|}\mathbb{E}\bigg[\mathbb{E}\Big[   \big( \hat{\mu}_j^{I_k^c}(\bm{X}_i) -\mu_j(\bm{X}_i) \big)^2  \cdot \big[  1-2\frac{\mathbb{E} \big[ \mathbbm{1}\{ Z_i=j\}\big| \bm{X}_i, I_k^c \big]}{e_j(\bm{X}_i)} + \frac{\mathbb{E} \big[\mathbbm{1}^2\{ Z_i=j\}\big| \bm{X}_i, I_k^c\big]}{e^2_j(\bm{X}_i)}  \big]\Big] \Big| I_k^c \bigg]&\\
%&=\frac{1}{|I_k|}\mathbb{E}\bigg[\mathbb{E}\Big[   \big( \hat{\mu}_j^{I_k^c}(\bm{X}_i) -\mu_j(\bm{X}_i) \big)^2  \cdot \big[  1-2\frac{\mathbb{P} \big[  Z_i=j \big| \bm{X}_i, I_k^c \big]}{e_j(\bm{X}_i)} + \frac{\mathbb{P} \big[ Z_i=j\big| \bm{X}_i, I_k^c \big]}{e^2_j(\bm{X}_i)}  \big]\Big] \Big| I_k^c \bigg]&\\
%&=\frac{1}{|I_k|}\mathbb{E}\bigg[\mathbb{E}\Big[   \big( \hat{\mu}_j^{I_k^c}(\bm{X}_i) -\mu_j(\bm{X}_i) \big)^2  \cdot \big( -1+ \frac{1}{e_j(\bm{X}_i)}   \big)\Big] \Big| I_k^c \bigg]&\\
&\leq \frac{1}{\xi|I_k|}\mathbb{E}\Bigg[   \left( \hat{\mu}_j^{I_k^c}(\bm{X}_i) -\mu_j(\bm{X}_i) \right)^2    \Bigg] = o_P\left(\frac{1}{n}\right).
%&=\frac{o(1)}{n}& (***)
\end{flalign*}
From above, we have shown that each term of summand in Equation (\ref{lam:equ1}) converges in propability to zero, and therefore, $\sqrt{n}(\hat{\tau}_{GCF,j,j'} - \tau_{GCF, j,j'}) \overset{p}{\rightarrow}0$. %Note that $\sqrt{n}(\hat{\tau}_{GCF,j,j'} - \tau_{GCF, j,j'})$ is also asymptotically normal by central limit theorem, since summand in Equation (\ref{lam:equ1}) are i.i.d. when conditioned on the out-of-fold sample $I_k^c$.
\end{proof}

Generally speaking, the convergence rates are relatively slow when using machine learning methods as plug-in estimators, especially when the true model is a simple linear regression. Theorem \ref{theorem:GDCDR} shows that when the machine learning model can consistently estimate the propensity scores and outcome regressions, and the risk decay converges in probability to zero, the estimator $\hat{\tau}_{GCF,j,j'}$ is not only consistent  but also converges to the true parameter at the fastest rate at $\sqrt{n}$. The sup-norm consistency and risk decay assumptions made in Theorem \ref{theorem:GDCDR} are relatively easy to be satisfied through flexible models. However, the estimator is still sensitive to the fundamental assumption of the potential outcomes framework: the positivity/overlap assumption, since the estimator is weighted by the inverse of the estimated propensity scores. Theorem \ref{theorem: var} computes the asymptotic variance of $\hat{\tau}_{GCF,j,j'}$ that attains the semiparametric efficiency bound for average treatment effect(ATE). %Obviously, sampling splitting will reduce the size of each fold, and in order to have behavior closer to the asymptotic assumptions, we may need a larger sample size or smaller number of $K$.

\begin{theorem}\label{theorem: var}
If the following two conditions hold, 
\begin{enumerate}
\item $\big|Y_i^{obs} \big| \leq M$,
\item $\sup_{x \in \mathcal{X}}\big| \hat{\mu}_j^{I_k^c}(\bm{x}) - \mu_j(\bm{x}) \big| \overset{p}{\rightarrow}0$ and  $\sup_{x \in \mathcal{X}}\big| \hat{e}_j^{I_k^c}(\bm{x}) - e_j(\bm{x}) \big| \overset{p}{\rightarrow}0$ $\forall j$,
\end{enumerate}
then 
\begin{equation}\label{V}
    \hat{V}ar(\hat{\tau}_{GCF, j,j'}) \overset{p}{\rightarrow} V 
= Var\bigg[ \mu_j(\bm{X}_i) - \mu_{j'}(\bm{X}_i) \bigg] + \mathbb{E} \left[ \frac{\sigma_j^2(\bm{X}_i)}{e_j(\bm{X}_i)} \right] + \mathbb{E} \left[ \frac{\sigma_{j'}^2(\bm{X}_i)}{e_{j'}(\bm{X}_i)} \right],
\end{equation}
where
\begin{equation*}
\sigma_j^2(\bm{X}_i)=\mathbb{E}\bigg[ \Big( Y_i^{obs} - \mu_j(\bm{X}_i)  \Big)^2 \bigg| \bm{X}_i \bigg].   \end{equation*}

\end{theorem}

\begin{proof}
%Without cross-fitting, the generalized AIPW estimator $\hat{\tau}_{GAIPW,j,j'}$ is \citep{tu2013using}:
%\begin{equation}
%\hat{\tau}_{GAIPW,j,j'} = \frac{1}{N}\sum_{i=1}^{N}\Big[ \hat{\mu}_j(\bm{X}_i) - \hat{\mu}_{j'}(\bm{X}_i)  + \mathbbm{1}\{ Z_i=j\}  \frac{Y_i^{obs}-\hat{\mu}_j(\bm{X}_i)   }{ \hat{e}_j(\bm{X}_i) } - \mathbbm{1}\{ Z_i=j'\}  \frac{Y_i^{obs} - \hat{\mu}_{j'}(\bm{X}_i)   }{ \hat{e}_{j'}(\bm{X}_i) }  \Big]
%\end{equation}

Define the oracle estimator (the true estimator) $\tau^{*}_{GAIPW,j,j'} = \tau^{*}_{GAIPW,j} - \tau^{*}_{GAIPW,j'},$
where 
\begin{equation*}
\tau^{*}_{GAIPW,j} = \frac{1}{N}\sum_{i=1}^{N}\Big[ \mu_j(\bm{X}_i) + \mathbbm{1}\{ Z_i=j\}  \frac{Y_i^{obs}-\mu_j(\bm{X}_i)   }{ e_j(\bm{X}_i) }  \Big].
\end{equation*}
First, we are going to show (\ref{oracle_Var}) holds:
\begin{equation}\label{oracle_Var}
Var(\tau^{*}_{GAIPW,j,j'}) \overset{p}{\rightarrow} V.
\end{equation}
%\begin{equation}\label{oracle_Var}
%Var(\tau^{*}_{GAIPW,j,j'}) \overset{p}{\rightarrow} Var\Bigg[ \mu_j(\bm{X}_i) - \mu_{j'}(\bm{X}_i) \Bigg] + \mathbb{E} \Big[ \frac{\sigma_j^2(\bm{X}_i)}{e_j(\bm{X}_i)} \Big] + \mathbb{E} \Big[ \frac{\sigma_{j'}^2(\bm{X}_i)}{e_{j'}(\bm{X}_i)} \Big].
%\end{equation}
Let $$T_i=\mu_j(\bm{X}_i) - \mu_{j'}(\bm{X}_i)  + \mathbbm{1}\{ Z_i=j\}  \frac{Y_i^{obs} - \mu_j(\bm{X}_i)   }{ e_j(\bm{X}_i) } - \mathbbm{1}\{ Z_i=j'\}  \frac{Y_i^{obs} - \mu_{j'}(\bm{X}_i)   }{ e_{j'}(\bm{X}_i) },$$ then we have
\begin{equation}
\frac{1}{N-1}\sum_{i=1}^{N}(T_i - \bar{T})^2 \overset{p}{\rightarrow} V_1,
\end{equation}
%\begin{equation}
%\frac{1}{N-1}\sum_{i=1}^{N}\Big[ \mu_j(\bm{X}_i) - \mu_{j'}(\bm{X}_i)  + \mathbbm{1}\{ Z_i=j\}  \frac{Y_i^{obs} - \mu_j(\bm{X}_i)   }{ e_j(\bm{X}_i) } - \mathbbm{1}\{ Z_i=j'\}  \frac{Y_i^{obs} - \mu_{j'}(\bm{X}_i)   }{ e_{j'}(\bm{X}_i) } - \hat{\tau}_{GDR,j,j'}  \Big]^2 
%\overset{p}{\rightarrow} V_1, \end{equation}
where $$V_1 = Var\Bigg[ \mu_j(\bm{X}_i) - \mu_{j'}(\bm{X}_i)  + \mathbbm{1}\{ Z_i=j\}  \frac{Y_i^{obs} - \mu_j(\bm{X}_i)   }{ e_j(\bm{X}_i) } - \mathbbm{1}\{ Z_i=j'\}  \frac{Y_i^{obs} - \mu_{j'}(\bm{X}_i)   }{ e_{j'}(\bm{X}_i) } \Bigg].$$
Notice that
\begin{flalign*}
& Var\Bigg[ \mathbbm{1}\{ Z_i=j\}  \frac{Y_i^{obs} - \mu_j(\bm{X}_i)   }{ e_j(\bm{X}_i) } \Bigg]
%&= \mathbb{E}\Bigg[ \mathbbm{1}^2\{ Z_i=j\} \Big[  \frac{Y_i^{obs} - \mu_j(\bm{X}_i)   }{ e_j(\bm{X}_i) } \Big]^2 \Bigg]&\\
= \mathbb{E}\Bigg[ \mathbb{E}\bigg[ \mathbbm{1}\{ Z_i=j\} \Big[  \frac{Y_i^{obs} - \mu_j(\bm{X}_i)   }{ e_j(\bm{X}_i) } \Big]^2 \bigg| \bm{X}_i \bigg] \Bigg]&\\
&= \mathbb{E}\Bigg[ \frac{1}{e^2_j(\bm{X}_i)} \cdot \mathbb{E}\bigg[ \mathbbm{1}\{ Z_i=j\} \bigg| \bm{X}_i \bigg] \cdot   \mathbb{E}\bigg[ \Big( Y_i^{obs} - \mu_j(\bm{X}_i)  \Big)^2 \bigg| \bm{X}_i \bigg] \Bigg]&\\
&= \mathbb{E}\Bigg[ \frac{1}{e^2_j(\bm{X}_i)} \cdot e_j(\bm{X}_i) \cdot   \mathbb{E}\bigg[ \Big( Y_i^{obs} - \mu_j(\bm{X}_i)  \Big)^2 \bigg| \bm{X}_i \bigg] \Bigg]=\mathbb{E} \Big[ \frac{\sigma_j^2(\bm{X}_i)}{e_j(\bm{X}_i)} \Big].
\end{flalign*} 
Then $V_1$ can be written as 
\begin{flalign*}
&V_1 %= Var\Bigg[ \mu_j(\bm{X}_i) - \mu_{j'}(\bm{X}_i)  + \mathbbm{1}\{ Z_i=j\}  \frac{Y_i^{obs} - \mu_j(\bm{X}_i)   }{ e_j(\bm{X}_i) } - \mathbbm{1}\{ Z_i=j'\}  \frac{Y_i^{obs} - \mu_{j'}(\bm{X}_i)   }{ e_{j'}(\bm{X}_i) } \Bigg]&\\
= Var\Bigg[ \mu_j(\bm{X}_i) - \mu_{j'}(\bm{X}_i) \Bigg] +  Var\Bigg[ \mathbbm{1}\{ Z_i=j\}  \frac{Y_i^{obs} - \mu_j(\bm{X}_i)   }{ e_j(\bm{X}_i) } \Bigg]+ Var\Bigg[\mathbbm{1}\{ Z_i=j'\}  \frac{Y_i^{obs} - \mu_{j'}(\bm{X}_i)   }{ e_{j'}(\bm{X}_i) } \Bigg]\\
&= Var\Bigg[ \mu_j(\bm{X}_i) - \mu_{j'}(\bm{X}_i) \Bigg] + \mathbb{E} \Big[ \frac{\sigma_j^2(\bm{X}_i)}{e_j(\bm{X}_i)} \Big] + \mathbb{E} \Big[ \frac{\sigma_{j'}^2(\bm{X}_i)}{e_{j'}(\bm{X}_i)} \Big]=V.
\end{flalign*}
Therefore, (\ref{oracle_Var}) holds. The variance of the oracle GAIPW estimator $\tau^{*}_{GAIPW,j,j'}$ is asymptotically efficient, which attains the semiparametric efficiency bound of average treatment effect (ATE) proved in literature \citep{hahn1998role}. Next, we are going to show that the variance of the estimator $\hat{\tau}_{GCF, j,j'}$ converges to that of the oracle estimator's and hence also asymptotically efficient.

We define a function mapping $t(\cdot): \Omega \mapsto E$, $t(i)=k$ if $i \in I_k$. Then the equation (\ref{eq:var_wald}) can be re-written as:
\begin{equation}
\begin{split}
&\hat{V}ar(\hat{\tau}_{GCF, j,j'}) = \frac{1}{N-1}\sum_{i=1}^N\Bigg( \hat{\mu}_j^{(-t(i))}(\bm{X}_i) - \hat{\mu}_{j'}^{(-t(i))}(\bm{X}_i) \nonumber\\
& + 
\mathbbm{1}\{ Z_i=j\}  \frac{Y_i^{obs}-\hat{\mu}_j^{(-t(i))}(\bm{X}_i) }{\hat{e}_j^{(-t(i))}(\bm{X}_i)} 
 - \mathbbm{1}\{ Z_i=j' \}  \frac{Y_i^{obs} - \hat{\mu}_{j'}^{(-t(i))}(\bm{X}_i) }{1- \hat{e}_{j'}^{(-t(i))}(\bm{X}_i)} - \hat{\tau}_{GCDR, j,j'}  \Bigg)^2.
\end{split}
\end{equation}

Let $$S_i= \hat{\mu}_j^{(-t(i))}(\bm{X}_i) - \hat{\mu}_{j'}^{(-t(i))}(\bm{X}_i)+ 
\mathbbm{1}\{ Z_i=j\}  \frac{Y_i^{obs}-\hat{\mu}_j^{(-t(i))}(\bm{X}_i) }{\hat{e}_j^{(-t(i))}(\bm{X}_i)}- \mathbbm{1}\{ Z_i=j' \}  \frac{Y_i^{obs} - \hat{\mu}_{j'}^{(-t(i))}(\bm{X}_i) }{1- \hat{e}_{j'}^{(-t(i))}(\bm{X}_i)}.$$
By Triangle Inequality, the absolute difference of $S_i$ and $T_i$ can be decomposed into three terms as follows, of each converges to zero in probability.
\begin{flalign}\label{lam:equ2-1}
|S_i-T_i| &\leq \Bigg|\mu_j(\bm{X}_i)-\hat{\mu}_j^{(-t(i))}(\bm{X}_i)\Bigg|+ \Bigg| \mu_{j'}(\bm{X}_i)-\hat{\mu}_{j'}^{(-t(i))}(\bm{X}_i) \Bigg| &\\\label{lam:equ2-2}
& +\mathbbm{1}\{ Z_i=j\} \cdot \Bigg|   \frac{Y_i^{obs} - \mu_j(\bm{X}_i)   }{ e_j(\bm{X}_i) } -  \frac{Y_i^{obs}-\hat{\mu}_j^{(-t(i))}(\bm{X}_i) }{\hat{e}_j^{(-t(i))}(\bm{X}_i)}\Bigg|&\\\label{lam:equ2-3}
&+ \mathbbm{1}\{ Z_i=j'\} \cdot \Bigg|  \frac{Y_i^{obs} - \mu_{j'}(\bm{X}_i)   }{ e_{j'}(\bm{X}_i) } -  \frac{Y_i^{obs}-\hat{\mu}_{j'}^{(-t(i))}(\bm{X}_i) }{\hat{e}_{j'}^{(-t(i))}(\bm{X}_i)}\Bigg|
\end{flalign}

The quantity (\ref{lam:equ2-1}) converges in probability to zero by assumption $b$ in the theorem. For the quantities in (\ref{lam:equ2-2}) and (\ref{lam:equ2-2}), it is suffice to show $\Bigg|  \frac{Y_i^{obs} - \mu_j(\bm{X}_i)   }{ e_j(\bm{X}_i) } -  \frac{Y_i^{obs}-\hat{\mu}_j^{(-t(i))}(\bm{X}_i) }{\hat{e}_j^{(-t(i))}(\bm{X}_i)}\Bigg| \overset{p}{\rightarrow} 0$. To show this, we apply the  Triangle Inequality again and obtain the followings:

\begin{flalign*}
& \Bigg|  \frac{Y_i^{obs} - \mu_j(\bm{X}_i)   }{ e_j(\bm{X}_i) } -  \frac{Y_i^{obs}-\hat{\mu}_j^{(-t(i))}(\bm{X}_i) }{\hat{e}_j^{(-t(i))}(\bm{X}_i)}\Bigg|&\\
& \leq \Bigg|  \frac{Y_i^{obs} - \mu_j(\bm{X}_i)   }{ e_j(\bm{X}_i) } -  \frac{Y_i^{obs} - \hat{\mu}_j^{(-t(i))}(\bm{X}_i)   }{ e_j(\bm{X}_i) } \Bigg| +  \Bigg| \frac{Y_i^{obs} - \hat{\mu}_j^{(-t(i))}(\bm{X}_i)   }{ e_j(\bm{X}_i) } - \frac{Y_i^{obs}-\hat{\mu}_j^{(-t(i))}(\bm{X}_i) }{\hat{e}_j^{(-t(i))}(\bm{X}_i)}\Bigg|&\\
& \leq \Bigg|   \frac{ \mu_j(\bm{X}_i) -\hat{\mu}_j^{(-t(i))}(\bm{X}_i)  }{ e_j(\bm{X}_i) } \Bigg| + \Bigg| \frac{ \Big[ Y_i^{obs}-\hat{\mu}_j^{(-t(i))}(\bm{X}_i) \Big] \cdot \Big[ e_j(\bm{X}_i) - \hat{e}_j^{(-t(i))}(\bm{X}_i) \Big] }{ e_j(\bm{X}_i) \cdot \hat{e}_j^{(-t(i))}(\bm{X}_i)}  \Bigg| &\\
& \leq \left|   \frac{ \mu_j(\bm{X}_i) -\hat{\mu}_j^{(-t(i))}(\bm{X}_i)  }{ \xi} \Bigg| + \Bigg| \frac{ \Big[ Y_i^{obs}-\hat{\mu}_j^{(-t(i))}(\bm{X}_i) \Big] \cdot \Big[ e_j(\bm{X}_i) - \hat{e}_j^{(-t(i))}(\bm{X}_i) \Big] }{\xi \cdot \xi}  \right|. 
\end{flalign*}
Therefore, by assumptions $a$ and $b$, $$\Bigg|  \frac{Y_i^{obs} - \mu_j(\bm{X}_i)   }{ e_j(\bm{X}_i) } -  \frac{Y_i^{obs}-\hat{\mu}_j^{(-t(i))}(\bm{X}_i) }{\hat{e}_j^{(-t(i))}(\bm{X}_i)}\Bigg| \overset{p}{\rightarrow} 0,$$ and hence $|S_i-T_i| \overset{p}{\rightarrow} 0$, which also implies $|\bar{S}-\bar{T}| \overset{p}{\rightarrow} 0$.

Next, notice that $\frac{1}{N-1} \sum_{i=1}^{N}(S_i - \bar{S})^2$ can be decomposed as follows:
\begin{flalign}
\frac{1}{N-1} \sum_{i=1}^{N}(S_i - \bar{S})^2 &=  \frac{1}{N-1}\sum_{i=1}^{N}(T_i - \bar{T})^2& \label{var_s_1}\\
&+ \frac{1}{N-1}\sum_{i=1}^{N}(S_i - T_i)^2 + \frac{N}{N-1}(\bar{T}-\bar{S})^2 \label{var_s_2}&\\
&+ \frac{2}{N-1}\sum_{i=1}^{N}(S_i - T_i)(T_i-\bar{T}) + \frac{2(\bar{T}-\bar{S})}{N-1}\sum_{i=1}^{N}(S_i - T_i)\label{var_s_3}.
\end{flalign}
Since (\ref{var_s_1}) $\overset{p}{\rightarrow} V$ and (\ref{var_s_2}), (\ref{var_s_3}) $\overset{p}{\rightarrow} 0$, then $\hat{V}ar(\hat{\tau}_{GCF, j,j'})=\frac{1}{N-1} \sum_{i=1}^{N}(S_i - \bar{S})^2 \overset{p}{\rightarrow} V$.
%$\frac{1}{N-1} \sum_{i=1}^{N}(S_i - \bar{S})^2
%= \frac{1}{N-1}\sum_{i=1}^{N}(S_i - T_i)^2 + \frac{1}{N-1}\sum_{i=1}^{N}(T_i - \bar{T})^2 + \frac{N}{N-1}(\bar{T}-\bar{S})^2 + \frac{2}{N-1}\sum_{i=1}^{N}(S_i - T_i)(T_i-\bar{T}) + \frac{2(\bar{T}-\bar{S})}{N-1}\sum_{i=1}^{N}(S_i - T_i)$ then  $\frac{1}{N-1} \sum_{i=1}^{N}(S_i - \bar{S})^2 \overset{p}{\rightarrow} V_1$
Therefore, we have shown that when the difference between $S_i$ and $T_i$ is small, the variance $\hat{V}ar(\hat{\tau}_{GCF, j,j'})\overset{p}{\rightarrow} V$, and hence is also asymptotically efficient.
\end{proof}

Next, we are going to show asymptotic normality of the proposed GCF estimators.
\begin{theorem}\label{theorem: normal}
Under assumptions in Theorem \ref{theorem:GDCDR}, the estimator $\hat{\tau}_{GCF, j,j'}$ is asymptotic normal, i.e. for $j \neq j' \text{ and } j,j' \in \{1,...,J\}$, $\sqrt{n}(\hat{\tau}_{GCF, j, j'} - \tau^{*}_{GCF,j,j'}) \overset{d}{\rightarrow} N(0, V)$, where 
\begin{align*}
\tau^{*}_{GCF, j,j'} &= \tau^{*}_{GCF, j}- \tau^{*}_{GCF, j'}, \tau^{*}_{GCF, j} = \sum_{k=1}^{K} \frac{|I_k|}{N} \tau^{* I_k}_{j}, \text{ and } \\ \tau^{* I_k}_{j} &= \frac{1}{|I_k|} \sum_{i \in I_k} \bigg[  \mu_j(\bm{X}_i) + 
\mathbbm{1}\{ Z_i=j\}  \frac{Y_i^{obs}-\mu_j(\bm{X}_i) }{e_j(\bm{X}_i)}  \bigg].
\end{align*}
%\begin{align}
%\tau^{* I_k}_{j,j'} = &\frac{1}{|I_k|} \sum_{i \in I_k} \Big[  \mu_j(\bm{X}_i) - \mu_{j'}(\bm{X}_i)
%+ 
%\mathbbm{1}\{ Z_i=j\}  \frac{Y_i^{obs}-\mu_j(\bm{X}_i) }{e_j(\bm{X}_i)} \nonumber \\ 
% &- \mathbbm{1}\{ Z_i=j' \}  \frac{Y_i^{obs} - \mu_{j'}(\bm{X}_i) }{1- e_{j'}(\bm{X}_i)} \Big]. \nonumber
%\end{align}
\end{theorem}

\begin{proof}%[Proof of Theorem \ref{theorem: normal}:]
Fix any $k \in \{1,...,K\}$, define a function $h\Big( \hat{\mu}_j^{I_k^c}(\bm{X}_i), \hat{e}_j^{I_k^c}(\bm{X}_i) \Big)$ as
\begin{align}
h\Big( \hat{\mu}_j^{I_k^c}(\bm{X}_i), \hat{e}_j^{I_k^c}(\bm{X}_i) \Big) &:= \hat{\mu}_j^{I_k^c}(\bm{X}_i) - \hat{\mu}_{j'}^{I_k^c}(\bm{X}_i) \nonumber \\ 
 &+ 
\mathbbm{1}\{ Z_i=j\}  \frac{Y_i^{obs}-\hat{\mu}_j^{I_k^c}(\bm{X}_i) }{\hat{e}_j^{I_k^c}(\bm{X}_i)} - \mathbbm{1}\{ Z_i=j' \}  \frac{Y_i^{obs} - \hat{\mu}_{j'}^{I_k^c}(\bm{X}_i) }{1- \hat{e}_{j'}^{I_k^c}(\bm{X}_i)}. \nonumber
\end{align}
Since when we have equal size for each $I_k$, $\hat{\tau}_{GCF, j,j'} =\frac{1}{K} \sum_{k=1}^{K} \hat{\tau}^{I_k}_{j,j'} = \hat{\tau}^{I_k}_{j,j'}$, then it is suffice to show that $\hat{\tau}^{I_k}_{j,j'} =  \frac{1}{|I_k|} \sum_{i \in I_k} h\Big( \hat{\mu}_j^{I_k^c}(\bm{X}_i), \hat{e}_j^{I_k^c}(\bm{X}_i) \Big)$ is normally distributed with the desired mean and variance.

Notice that $\forall$ $i \neq t \in \{1,...,N\}$ such that $X_i \in I_k$ and $X_t \in I_k$, since $X_i$'s are i.i.d, we have 
\begin{equation*}
    h\Big( \hat{\mu}_j^{I_k^c}(\bm{X}_i), \hat{e}_j^{I_k^c}(\bm{X}_i) \Big) \perp h\Big( \hat{\mu}_j^{I_k^c}(\bm{X}_t), \hat{e}_j^{I_k^c}(\bm{X}_t) \Big).
\end{equation*}
Since we have proved $\sqrt{n}(\hat{\tau}^{I_k}_{j,j'} - \tau^{* I_k}_{j,j'}) \overset{p}{\rightarrow} 0$ and $\hat{V}ar(\hat{\tau}^{I_k}_{j,j'})=\hat{V}ar(\hat{\tau}_{GCF,j,j'}) \overset{p}{\rightarrow} V$ in the proofs of Theorem \ref{theorem:GDCDR} and Theorem \ref{theorem: var}, using the central limit theorem, we obtain $\sqrt{n}(\hat{\tau}^{I_k}_{j,j'} - \tau^{* I_k}_{j,j'}) \overset{d}{\rightarrow} N(0, V)$.
%Use the same function mapping $t(\cdot): \Omega \mapsto E$, $t(i)=k$ if $i \in I_k$ defined in the proof of Theorem \ref{theorem: var}, $\hat{\tau}_{GCF, j, j'}$ can be re-written as:
%\begin{equation*}
%\begin{split}
%\hat{\tau}_{GCF, j,j'} &= \frac{1}{N}\sum_{i=1}^{N}\Bigg[ \hat{\mu}_j^{(-t(i))}(\bm{X}_i) - \hat{\mu}_{j'}^{(-t(i))}(\bm{X}_i) &\\
%& + 
%\mathbbm{1}\{ Z_i=j\}  \frac{Y_i^{obs}-\hat{\mu}_j^{(-t(i))}(\bm{X}_i) }{\hat{e}_j^{(-t(i))}(\bm{X}_i)} 
% - \mathbbm{1}\{ Z_i=j' \}  \frac{Y_i^{obs} - \hat{\mu}_{j'}^{(-t(i))}(\bm{X}_i) }{1- \hat{e}_j^{(-t(i))}(\bm{X}_i)} \Bigg]&
%\end{split}
%\end{equation*}
\end{proof}

Therefore, we have $\sqrt{n}(\hat{\tau}_{GCF, j, j'} - \tau_{GCF,j,j'}) \overset{d}{\rightarrow} N(0, V)$. The confidence intervals constructed in (\ref{eq:var_wald}) is valid for large sample.

\subsection{Discussions and Comments}
The theoretical results of multi-treatment comparison using generalized doubly robust estimator with cross-fitting is similar to the binary case. The estimator $\hat{\tau}_{GCF, j,j'}$ is still doubly robust in the sense that if at least one of the postulated models is correctly specified, $\hat{\tau}_{GCF, j,j'}$ is a consistant estimator of $\tau_{ATE, j,j'}$. Without using cross-fitting approach, one can directly plug $\hat{\mu}_j^{I_k^c}(\bm{x})$ and $\hat{e}_j^{I_k^c}(\bm{x})$ estimated based on the entire dataset into the doubly robust estimator. Such approach requires restricting the model of $\mu_j(\bm{x})$ and $e_j(\bm{x})$ to be within Donsker function class so as to compute the asymptotic efficiency bound through deriving the corresponding efficient influence function expansion based on independent samples. In practice, it might be hard to show a machine learning model satisfy the Donsker class, and hence in theory, a reliable statistical inference can not be guaranteed. 

On the contrary, by using cross-fitting approach, one can make relatively weak assumptions that both root-mean squared errors of the nuisance parameters are bounded such that $MSE\big[ \hat{\mu}_j^{I_k^c}(\bm{x})\big] =o_p(n^{-1/4})$ and $MSE\big[ \hat{e}_j^{I_k^c}(\bm{x})  \big]=o_p(n^{-1/4})$, $\forall j \in \{1,...,J\}$. Under these assumptions, we have showed in Theorems \ref{theorem:GDCDR} and \ref{theorem: var} that a wider range of machine learning models can be applied to estimate the nuisance parameters and the GCF estimator still achieves $\sqrt{n}$ convergence rate and asymptotic efficiency.

The asymptotic results of GCF estimator provide a valid calculation for simultaneous confidence intervals and the variance estimation attains the semiparametric efficiency bound for average treatment effect (ATE). In other words, the variance estimation is optimal such that no "regular" estimators can improve its asymptotic performance. Additionally, similar to the idea of cross-validation in estimating test errors of an estimator, another side benefit of using cross-fitting is to mitigate overfitting in the case of over-parameterization or non-parametric estimation of $\mu_j(\bm{x})$ and $e_j(\bm{x})$.

\newpage
\section{Simulation Studies}
\subsection{Setups}
The performance of the generalized doubly robust estimator with cross-fitting (GCF) is further demonstrated and evaluated by comparing with other existing estimators for multiple treatment comparisons. We focus on comparing the estimator with the other two estimators difference-in-means (DIF) and generalized-AIPW (GAIPW) under three treatment levels ($J=3$) and six treatment levels ($J=6$). For each of the estimator, we access the bias, RMSE (Root-mean-square-error), and coverage of the 95\% confidence intervals. Bias is defined as the difference between estimated ATE and the true ATE. RMSE is calculated using the square root of the mean of squared bias. All the 95\% confidence intervals are computed using the wald-type asymptotic confidence interval. The coverage of confidence interval is calculated as the proportion of intervals containing the true parameter. For the GCF estimator, we choose fold size as three in estimating the nuisance parameters.

The data generating mechanism is similar to \citet{yang2016propensity}. Each subject is associated with a covariate vector $\bm{X}_i^T=(X_{i1}, X_{i2}, X_{i3}, X_{i4}, X_{i5}, X_{i6})$, where
\begin{equation*}
    \begin{bmatrix}
X_{i1}\\
X_{i2}\\
X_{i3}
\end{bmatrix} \sim N\left( \begin{bmatrix}
2\\
1\\
1
\end{bmatrix}, \begin{bmatrix}
2 & 1 & -1\\
1 & 1 & -0.5\\
-1 & -0.5  & 1
\end{bmatrix} \right),
\end{equation*}
$X_{i4} \sim U[-3,3]$, $X_{i5} \sim \chi_1^2$, $X_{i6} \sim Ber(0.5)$. The propensity scores are generated via multinomial logistic regression \begin{equation*}
    (D_{i1},...,D_{iJ}) \big| \bm{X}_i \sim Multinom\big(e_1(\bm{X}_i),...,e_J(\bm{X}_i) \big),
\end{equation*} where $e_j(\bm{X}_i) = \frac{\exp{(\bm{X}_i^T\bm{\beta}_j)}}{ \sum_{t=1}^{J} \exp{(\bm{X}_i^T\bm{\beta}_t})}$. The potential outcomes are from $Y_i(j) = \bm{X}_i^T\bm{\alpha}_j + \epsilon_i$, where the noise term $\epsilon_i \sim N(0,1)$. 

In our first and second simulation designs, we consider three treatment groups comparisons with sample size fixed at $N=1500$, $N=4500$, and $N=6000$. For the first simulation, we choose $\bm{\alpha}_1^T=(-1.5,1,1,1,1,1,1)$, $\bm{\alpha}_2^T=(-2, 2, 3, 1, 2, 2, 2)$, $\bm{\alpha}_3^T=(2, 3, 1, 2, -1, -1, -1)$, $\bm{\beta}_1^T=(0,0,0,0,0,0,0)$, $\bm{\beta}_2^T=0.3\times(0,1,1,1,-1,1,1)$, $\bm{\beta}_3^T=0.1\times(0,1,1,1,1,1,1)$  to simulate adequate overlap. The second simulation design also considers the same $\bm{\alpha}$'s as the first simulation design, while we choose $\bm{\beta}_1^T=(0,0,0,0,0,0,0)$, $\bm{\beta}_2^T=0.9\times(0,1,1,1,-1,1,1)$, $\bm{\beta}_3^T=0.1\times(0,1,1,1,1,1,1)$ to simulate lack of overlap. In our third simulation design, we consider six treatment groups comparisons with sample size fixed at $N=1500$, $N=4500$, and $N=6000$. We choose $\bm{\alpha}_1^T=(-1.5, 1, 1, 1, 1, 1, 1)$, $\bm{\alpha}_2^T=(-3, 2, 3, 1, 2, 2, 2)$, $\bm{\alpha}_3^T=(3, 3, 1, 2, -1, -1, 4)$, $\bm{\alpha}_4^T=(2.5, 4, 1, 2, -1, -1, -3)$, $\bm{\alpha}_5^T=(2, 5, 1, 2, -1, -1, -2)$ $\bm{\alpha}_5^T=(1.5, 6, 1, 2, -1, -1, -1)$, $\bm{\beta}_1^T=(0, 0, 0, 0, 0, 0, 0)$, $\bm{\beta}_2^T=0.2\times(0, 1, 1, 2, 1, 1, 1)$, $\bm{\beta}_3^T=0.3\times(0, 1, 1, 1, 1, 1, -5)$, $\bm{\beta}_3^T=0.4\times(0, 1, 1, 1, 1, 1, 5)$, $\bm{\beta}_4^T=0.5\times(0, 1, 1, 1, -2, 1, 1)$, $\bm{\beta}_4^T=0.6\times(0, 1, 1, 1, -2, -1, 1)$. In each simulation, the estimation repeats 2000 times.

\subsection{Results}
For the first and second simulation designs with $J=3$, we present the simulation results of sample size fixed at $N=1500$, $N=4500$, and $N=6000$ for both scenarios, adequate covariate overlap and lack of covariate overlap, in Table \ref{sim_res2} to Table \ref{sim_res7}. As can be seen from the results below, the performance of the baseline estimator DIF is significantly worse than the other two estimators, GAIPW and GCF, in either scenario. For the GAIPW and GCF estimators, when covariates have adequate overlap, both estimators can unbiasly estimate the average treatment effect for each treatment group comparison, and the performances of RMSEs and coverages are similar. As expected, the RMSEs are small and the coverage of confidence intervals are around $95\%$. On the other hand, with lack of overlap, the Table \ref{sim_res5} shows that due to the extreme propensity scores, the estimates of the average treatment effects are biased and under-coverage. 

For the third simulation design with $J=6$, we present the simulation results in Figure \ref{fig:cf_1500_6.1} to Figure \ref{fig:cf_1500_6.3}. As can be seen from Figure \ref{fig:cf_1500_6.1}, when the total sample size is small ($N=1500$), the RMSE and coverage of the GCF estimator are not as good as the GAIPW estimator, which shows the downside of using GCF estimator. When fitting nuisance parameters, GCF estimator divides the data into $K$ folds and fit the models with only one of the folds each time. It is not problematic when sample size is large. However, if the sample size is small, data splitting will result in numerical instability due to inadequate amount of data in fitting the nuisance models.

The simulation results for multiple treatment comparisons are consistent with the simulation results for binary settings presented by \citet{zivich2021machine}, which shows the performance of the GCF estimator were as efficient as the GAIPW estimator in binary setting. When $J=6$, both GCF and GAIPW estimators are unbiased in estimating average treatment effects, while the $95\%$ confidence interval coverages are not as good as $J=3$. This might due to the overlap issues are magnified when the number of treatment groups increases. The performance of the GCF estimator is similar to the performance of the estimator without cross-fitting when the nuisance models are correctly specified, the theorems proved in Chapter 3.2.2 provided both transparent assumptions and theories to support rigorous statistical inference while sufficiently flexible machine learning tools are used to model propensity scores and outcome regressions. High-dimensional data are common in practice, which might rely on complex machine learning tools to model the nuisance models. The GCF estimator plays as a safeguard to perform valid inference without concerning whether the models are from Donsker class or not.

\begin{table}[h!]
\begin{center}
	\caption{Simulation results of adequate overlap for J = 3 treatment group comparison with sample size $N=1500$.}
	\label{sim_res2}
	\begin{tabular}{l|llll}
		\hline
		Metric   &       & $\tau_{1, 2}$ & $\tau_{1, 3}$ & $\tau_{2,3}$ \\
		\hline 
		Bias     & DIF   &  0.9018 & -0.2340 & -1.1359   \\
		& GAIPW &  0.0021 & 0.0042 & 0.0021  \\
		& GCF   & 0.0019 & 0.0047 & 0.0028 \\
		\hline
		RMSE     & DIF   & 0.9546  & 0.3491 & 1.1878    \\
		& GAIPW & 0.1243  & 0.1502 & 0.1969   \\
		& GCF   & 0.1286  & 0.1528 & 0.1970   \\
		\hline
		Coverage & DIF   & 16.80\% & 84.00 \% & 9.15\%      \\
		& GAIPW & 94.60\% & 94.55\% & 95.10\%      \\
		& GCF   & 94.45\% & 94.05\% & 95.00\%      \\
		\hline
	\end{tabular}
\end{center}
\end{table}

\begin{table}[h!]
\begin{center}
	\caption{Simulation results of adequate overlap for J = 3 treatment group comparison with sample size $N=4500$.}
	\label{sim_res3}
	\begin{tabular}{l|llll}
		\hline
		Metric   &       & $\tau_{1, 2}$ & $\tau_{1, 3}$ & $\tau_{2,3}$ \\
		\hline 
		Bias     & DIF   & 0.9070 & -0.2403 & -1.1472  \\
		& GAIPW & 0.0008 &  0.0004 & -0.0004  \\
		& GCF   & 0.0004 & -0.0002 & -0.0004 \\
		\hline
		RMSE     & DIF   & 0.9244 & 0.2812 & 1.1643   \\
		& GAIPW & 0.0720 & 0.0850 & 0.1094 \\
		& GCF   & 0.0757 & 0.0888 & 0.1094 \\
		\hline
		Coverage & DIF   & 0.00\%    &  63.15\%    & 0.00\%    \\
		& GAIPW & 93.60\%    & 94.50\%    & 95.45\%  \\
		& GCF   & 93.85\%    & 94.15\%    & 95.45\%  \\
		\hline
	\end{tabular}
\end{center}	
\end{table}

\begin{table}[h!]
\begin{center}
	\caption{Simulation results of adequate overlap for J = 3 treatment group comparison with sample size $N=6000$.}
	\label{sim_res4}
	\begin{tabular}{l|llll}
		\hline
		Metric   &       & $\tau_{1, 2}$ & $\tau_{1, 3}$ & $\tau_{2,3}$ \\
		\hline 
		Bias     & DIF   & 0.9057  & -0.2442 & -1.1499  \\
		& GAIPW & -0.0014 & 0.0023 & 0.0037  \\
		& GCF   & -0.0015 & 0.0022 & 0.0037 \\
		\hline
		RMSE     & DIF   & 0.9193 & 0.2747 & 1.1628   \\
		& GAIPW & 0.0615 & 0.0764 & 0.0973 \\
		& GCF   & 0.0621 & 0.0769 & 0.0973 \\
		\hline
		Coverage & DIF   & 0.00\%    &  52.10\%    & 0.00\%    \\
		& GAIPW & 94.15\%    & 93.80\%    & 95.30\%   \\
		& GCF   & 93.85\%  & 93.90\%      & 95.35\%  \\
		\hline
	\end{tabular}
\end{center}
\end{table}

\begin{table}[h!]
\begin{center}
	\caption{Simulation results of lack of overlap for J = 3 treatment group comparison with sample size $N=1500$.}
	\label{sim_res5}
	\begin{tabular}{l|llll}
		\hline
		Metric   &       & $\tau_{1, 2}$ & $\tau_{1, 3}$ & $\tau_{2,3}$ \\
		\hline 
		Bias     & DIF   & 1.4518  & -1.3291 & -2.7809  \\
		& GAIPW & 0.0023 & -0.0051 & -0.0074 \\
		& GCF   & -0.0039  & -0.0877  & -0.0838  \\
		\hline
		RMSE     & DIF   & 1.4807  &  1.3555  &  2.7999   \\
		& GAIPW & 0.2417 &  0.3608 &  0.3312   \\
		& GCF   & 0.3595 & 3.4929 & 3.4846 \\
		\hline
		Coverage & DIF   & 0.00\%       & 0.15\%      & 0.00\%      \\
		& GAIPW & 78.65\%       &  75.85\%       &  87.40\%      \\
		& GCF   & 74.55\%       & 69.00\%       & 84.00\%      \\
		\hline
	\end{tabular}
\end{center}
\end{table}

\begin{table}[h!]
\begin{center}
	\caption{Simulation results of lack of overlap for J = 3 treatment group comparison with sample size $N=4500$.}
	\label{sim_res6}
	\begin{tabular}{l|llll}
		\hline
		Metric   &       & $\tau_{1, 2}$ & $\tau_{1, 3}$ & $\tau_{2,3}$ \\
		\hline 
		Bias     & DIF   & 1.4605 & -1.3302 & -2.7908  \\
		& GAIPW & 0.0054  & -0.0052  &  -0.0106 \\
		& GCF   & 0.0067 & -0.0045 & -0.0112 \\
		\hline
		RMSE     & DIF   & 1.4700   &  1.3391  &  2.7966    \\
		& GAIPW & 0.2072 & 0.2337 & 0.1612 \\
		& GCF   & 0.2731 & 0.3039 & 0.1843\\
		\hline
		Coverage & DIF   & 0.00\%       & 0.00\%      & 0.00\%      \\
		& GAIPW & 74.10\%       & 72.35\%       & 88.00\%      \\
		& GCF   &  72.15\%       &  69.40\%       &  86.55\%      \\
		\hline
	\end{tabular}
\end{center}
\end{table}

\begin{table}[h!]
\begin{center}
	\caption{Simulation results of lack of overlap for J = 3 treatment group comparison with sample size $N=6000$.}
	\label{sim_res7}
	\begin{tabular}{l|llll}
		\hline
		Metric   &       & $\tau_{1, 2}$ & $\tau_{1, 3}$ & $\tau_{2,3}$ \\
		\hline 
		Bias     & DIF   & 1.4596  &  -1.3334 & -2.7930  \\
		& GAIPW & 0.0038  &  0.0055  &  0.0017 \\
		& GCF   & 0.0064   & 0.0104  & 0.0040 \\
		\hline
		RMSE     & DIF   & 1.4667   &  1.3400 &  2.7975   \\
		& GAIPW & 0.1633  &  0.2145  & 0.1772 \\
		& GCF   & 0.1748  &  0.2688 & 0.2302\\
		\hline
		Coverage & DIF   & 0.00\%       & 0.00\%      & 0.00\%      \\
		& GAIPW & 73.65\%       & 71.85\%       & 86.95\%      \\
		& GCF   &  72.60\%       & 69.70\%       & 86.10\%      \\
		\hline
	\end{tabular}
\end{center}
\end{table}

\begin{figure}[h!]
\minipage{0.32\textwidth}
  \includegraphics[width=\linewidth]{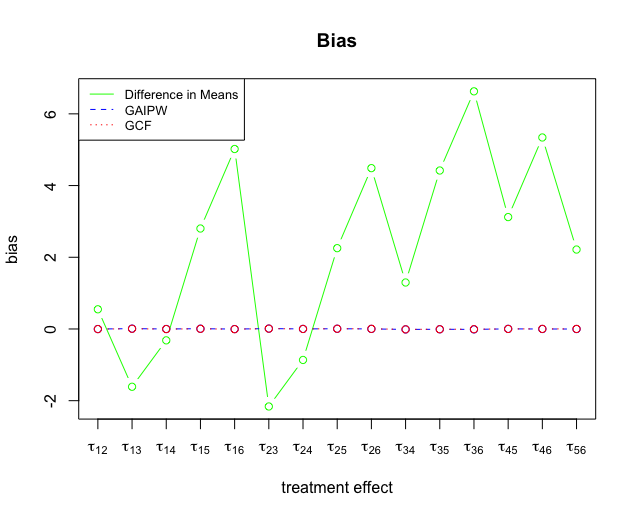}
\endminipage\hfill
\minipage{0.32\textwidth}
  \includegraphics[width=\linewidth]{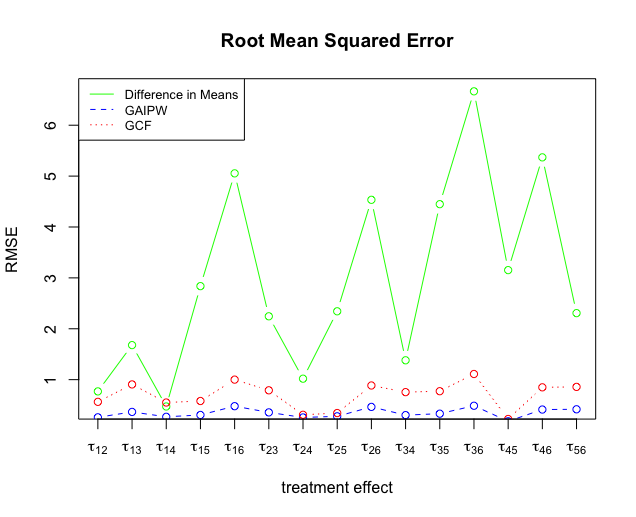}
\endminipage\hfill
\minipage{0.32\textwidth}%
  \includegraphics[width=\linewidth]{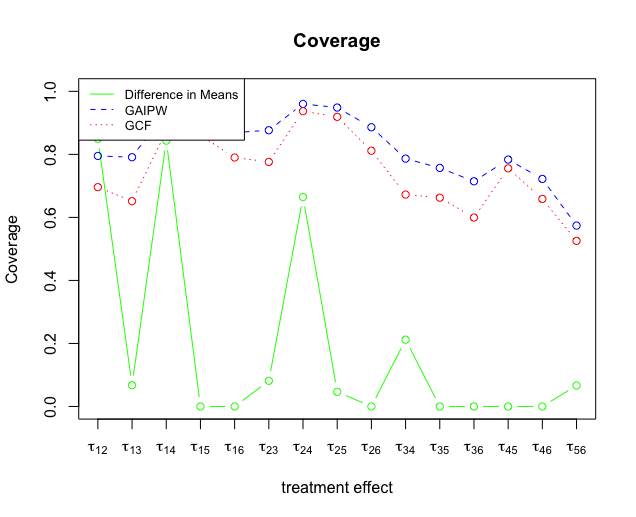}
\endminipage
\caption{Simulation results for J = 6 treatment group comparison with sample size $N=1500$.}\label{fig:cf_1500_6.1}
\end{figure}

\begin{figure}[h!]
\minipage{0.32\textwidth}
  \includegraphics[width=\linewidth]{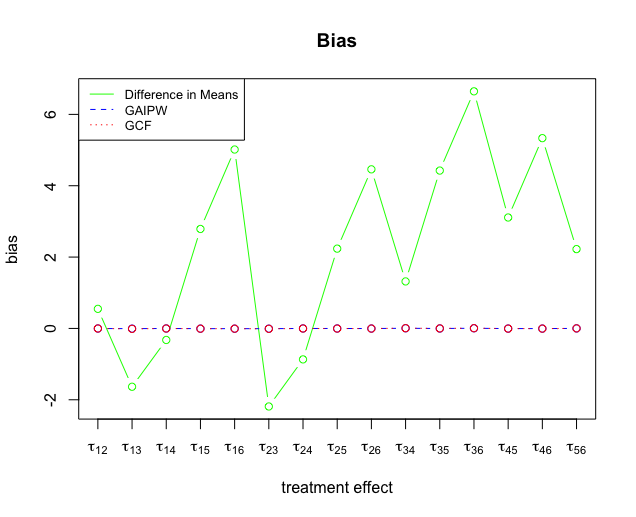}
\endminipage\hfill
\minipage{0.32\textwidth}
  \includegraphics[width=\linewidth]{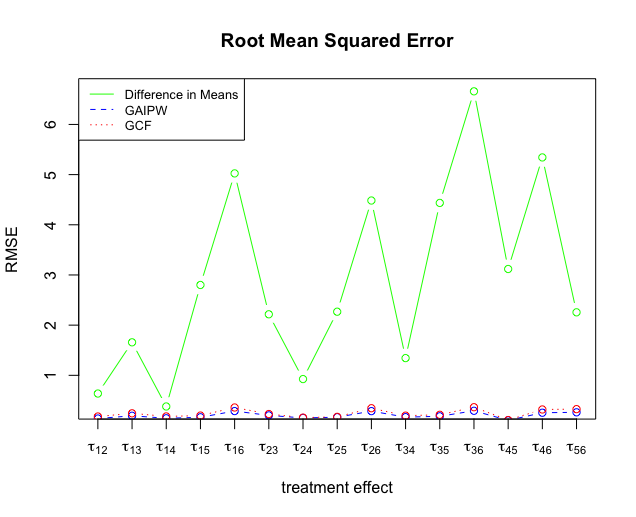}
\endminipage\hfill
\minipage{0.32\textwidth}%
  \includegraphics[width=\linewidth]{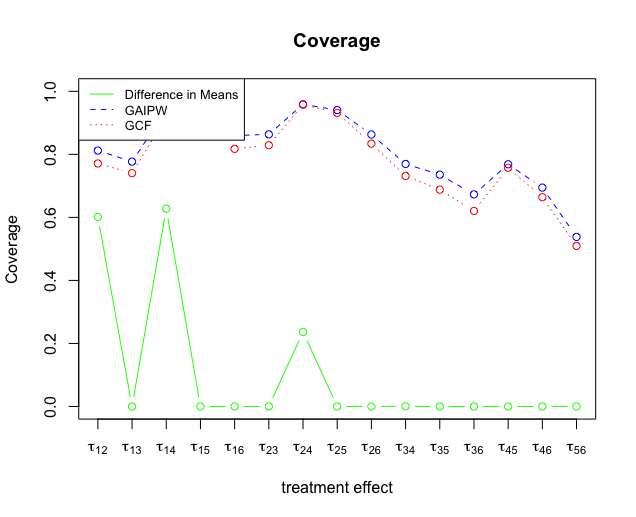}
\endminipage
\caption{Simulation results for J = 6 treatment group comparison with sample size $N=4500$.}\label{fig:cf_1500_6.2}
\end{figure}

\begin{figure}[h!]
\minipage{0.32\textwidth}
  \includegraphics[width=\linewidth]{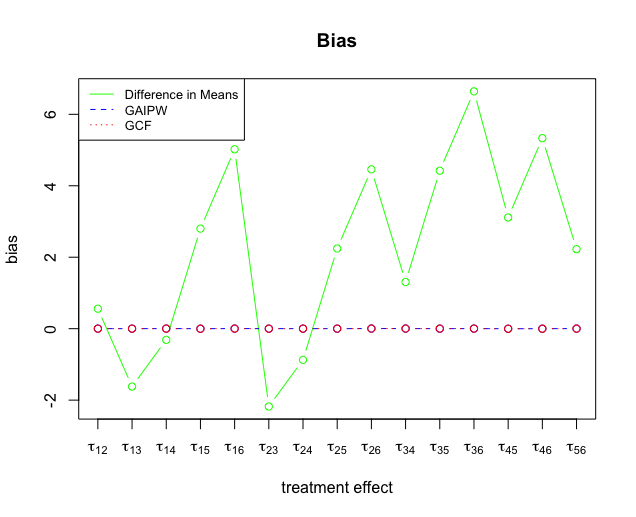}
\endminipage\hfill
\minipage{0.32\textwidth}
  \includegraphics[width=\linewidth]{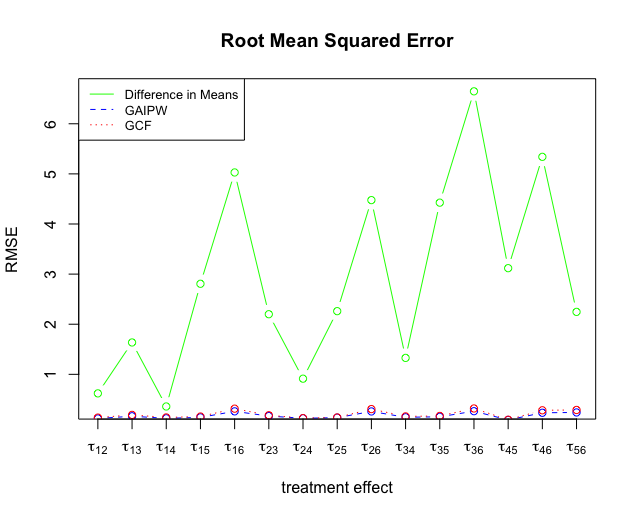}
\endminipage\hfill
\minipage{0.32\textwidth}%
  \includegraphics[width=\linewidth]{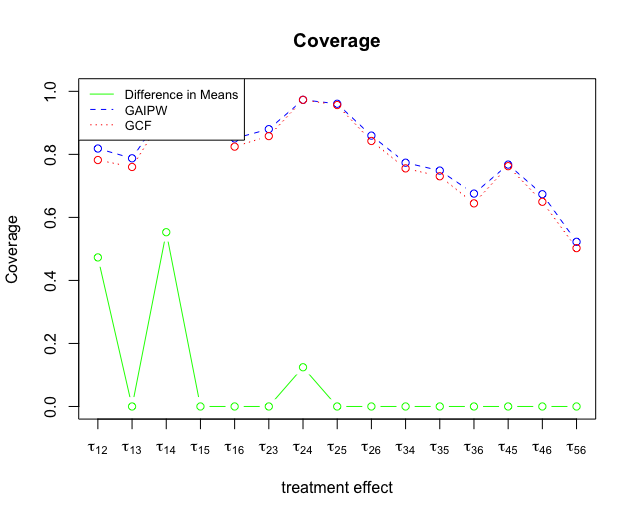}
\endminipage
\caption{Simulation results for J = 6 treatment group comparison with sample size $N=6000$.}\label{fig:cf_1500_6.3}
\end{figure}

\newpage
\section{Discussion and Conclusion}

In this paper, we generalized doubly-robust estimator with cross-fitting in multiple treatment comparison settings and provided the asymptotic simultaneous confidence intervals that attains semiparametric efficiency bound of average treatment effect. The asymptotic behaviors of this estimator have also been investigated. 

When simple parametric models do not fit the nuisance parameters well, one might rely on more flexible models to improve the goodness of the fit. The major benefit of using GCF estimator is that the cross-fitting technique provides comparably simpler assumptions than requiring models to be from Donsker class, so that flexible machine learning models can be applied in modeling the nuisance parameters and the asymptotic properties of the estimator such as  $\sqrt{n}$-consistency, asymptotic normality, and asymptotic efficiency are still promised. 

For GCF estimator, the assumptions on the consistency of outcome regressions and propensity scores, and on the convergence speed are relatively easy to satisfy through flexible machine learning models. However, similar to other estimators derived based on potential outcomes framework, the overlap assumption is still a strong assumption to make, especially in multiple treatment settings. Also, sampling splitting reduces the size of each fold, and in order to have behavior closer to the asymptotic assumptions, one may need a larger sample size or smaller number of $K$. The other cost of using cross-fitting techniques in estimating average treatment effects is that with small sample size, sample splitting might cause numerical instability. In the next chapter, we are going to discuss more complex treatment intervention settings, when taking time-varying variables into consideration.

\bibliographystyle{apalike}
\bibliography{thesis}
\end{document}